\documentclass[journal,twoside,web]{ieeecolor}
\usepackage{amsfonts}       % blackboard math symbols
\usepackage{nicefrac}       % compact symbols for 1/2, etc.
\usepackage{microtype}      % microtypography
\usepackage{bm}
\usepackage{float}
\usepackage{algorithm}
\usepackage{algorithmic}
\usepackage{color}

\usepackage{mathtools}
\usepackage{generic}

\usepackage{amsthm}
\usepackage{amssymb}
\usepackage{amsmath}
\usepackage{bbm}
\usepackage{mathtools}
\usepackage{breqn}
\usepackage{etex}
\usepackage{tikz}
\usepackage{tikzsymbols}
\usetikzlibrary{calc,shapes,arrows,backgrounds,arrows.meta,positioning}
\usepackage{pgfplots}
\pgfplotsset{compat=newest}
\pgfplotsset{plot coordinates/math parser=false}
\newlength\figureheight
\newlength\figurewidth
\usepackage{dashrule}
\usepackage{xpatch}
\xapptocmd{\appendix}{%
  \counterwithin{theorem}{subsection}
}{\typeout{Success}}{}
\newtheorem{assumption}{Assumption}
\theoremstyle{plain}
\newtheorem{theorem}{Theorem}[section]
\theoremstyle{plain}

\theoremstyle{plain}
\theoremstyle{definition}
\newtheorem{definition}[theorem]{Definition}
\theoremstyle{plain}
\newtheorem{lemma}[theorem]{Lemma}
\theoremstyle{plain}

\theoremstyle{plain}
\newtheorem{remark}[theorem]{Remark}
\newcommand{\RR}{\mathbb{R}}

\newcommand{\mSetPosSymMat}[1]{S^{#1}_{++}}

\newcommand{\mNormGen}[1]{\left\lVert {#1} \right\rVert}
\newcommand{\mNormGenSmall}[1]{\lVert {#1} \rVert}
\newcommand{\mDefFunction}[3]{#1: #2 \rightarrow #3}

\newcommand{\mDef}{\coloneqq}
\newcommand{\mOnes}[1]{\mathbbm 1_{#1}}

\newcommand{\mH}{h_{\mathrm{PB}}}

\newcommand{\mUpred}{u}
\newcommand{\mUpredOpt}{u^*}
\newcommand{\mUpredCand}{\bar u}
\newcommand{\mXpred}{x}
\newcommand{\mXpredOpt}{x^*}
\newcommand{\mXpredCand}{\bar x}

\newcommand{\XX}{\mathbb{X}}
\newcommand{\UU}{\mathbb{U}}
\renewcommand{\SS}{\mathcal{S}}
\newcommand{\DD}{\mathcal{D}}

\newcommand{\NN}{\mathbb{N}}

\newcommand{\KK}{\mathcal{K}}

\newcommand{\BB}{\mathcal{B}}
\renewcommand{\AA}{\mathcal A}

\newcommand{\ZZ}{\mathbb{Z}}

\renewcommand{\BB}{\mathcal{B}}

\DeclareMathOperator*{\argmin}{argmin}

\newcommand*{\END}{\hfill\ensuremath{\lhd}}

\tikzstyle{block} = [draw, rectangle,
minimum height=1cm, minimum width=3cm, text width=2.75cm, align=center]
\tikzstyle{sum} = [draw, fill=blue!20, circle, node distance=1cm]
\tikzstyle{input} = [coordinate]
\tikzstyle{output} = [coordinate]
\tikzstyle{pinstyle} = [pin edge={to-,thin,black}]
\makeatletter
\newif\ifmygrid@coordinates
\tikzset{/mygrid/step line/.style={line width=0.80pt,draw=gray!80},
	/mygrid/steplet line/.style={line width=0.25pt,draw=gray!80}}
\pgfkeys{/mygrid/.cd,
	step/.store in=\mygrid@step,
	steplet/.store in=\mygrid@steplet,
	coordinates/.is if=mygrid@coordinates}
\def\mygrid@def@coordinates(#1,#2)(#3,#4){%
	\def\mygrid@xlo{#1}%
	\def\mygrid@xhi{#3}%
	\def\mygrid@ylo{#2}%
	\def\mygrid@yhi{#4}%
}
\newcommand\DrawGrid[3][]{%
	\pgfkeys{/mygrid/.cd,coordinates=true,step=1,steplet=0.2,#1}%
	\draw[/mygrid/steplet line] #2 grid[step=\mygrid@steplet] #3;
	\draw[/mygrid/step line] #2 grid[step=\mygrid@step] #3;
	\mygrid@def@coordinates#2#3%
	\ifmygrid@coordinates%
		\draw[/mygrid/step line]
		\foreach \xpos in {\mygrid@xlo,...,\mygrid@xhi} {%
				(\xpos,\mygrid@ylo) -- ++(0,-3pt)
				node[anchor=north] {$\xpos$}
			}
		\foreach \ypos in {\mygrid@ylo,...,\mygrid@yhi} {%
				(\mygrid@xlo,\ypos) -- ++(-3pt,0)
				node[anchor=east] {$\ypos$}
			};
	\fi%
}
\makeatother

\begin{document}

\title{Predictive control barrier functions: Enhanced safety mechanisms for learning-based control}

\author{Kim P.~Wabersich and Melanie~N.~Zeilinger % <-this % stops a space
\thanks{The authors are with the Institute for Dynamic Systems and Control, ETH Z\"urich, Z\"urich
CH-8092, Switzerland (e-mail: \{wkim$\vert$mzeilinger\}@ethz.ch)}% <-this % stops a space
\thanks{This work was supported by the Swiss National Science Foundation under grant no. PP00P2 157601/1.}}

% make the title area
\maketitle

\begin{abstract}
	While learning-based control techniques often outperform classical controller designs, safety requirements limit the acceptance of such methods in many applications. Recent developments address this issue through so-called predictive safety filters, which assess if a proposed learning-based control input can lead to constraint violations and modifies it if necessary to ensure safety for all future time steps. The theoretical guarantees of such predictive safety filters rely on the model assumptions and minor deviations can lead to failure of the filter putting the system at risk. This paper introduces an auxiliary soft-constrained predictive control problem that is always feasible at each time step and asymptotically stabilizes the feasible set of the original predictive safety filter problem, thereby providing a recovery mechanism in safety-critical situations. This is achieved by a simple constraint tightening in combination with a terminal control barrier function. By extending discrete-time control barrier function theory, we establish that the proposed auxiliary problem provides a `predictive' control barrier function. The resulting algorithm is demonstrated using numerical examples. 
\end{abstract}


\section{INTRODUCTION}	
	The increasing availability of computational resources opens new perspectives for
	control engineering, and in particular enables learning-based control, which has
	shown the potential of solving complex high-level tasks.
	Demonstrations include, e.g., human-machine	interactions, where solely a
	`black-box' reward signal describes the desired system behavior.
	As such a general problem formulation is not addressed by classical controller
	specifications in terms of stability w.r.t. pre-specified setpoints or reference trajectories,
 	there is a renewed interest in the development of universal mechanisms to ensure safety
 	of the resulting closed-loop system.
	\par
	Using a modular approach to control system safety and performance as proposed
	in~\cite{Seto1998}, various methods have recently been developed, such as, e.g., control
	barrier functions~\cite{Ames2019} and reachability analysis based safety
	frameworks~\cite{Chen2018} to ensure safety in the form of constraint satisfaction,
	independent of a specific control task.
	These methods typically consist of a \emph{safety controller} that renders a \emph{safe subset}
	of the system constraints invariant. These two components allow to enhance an arbitrary
	performance controller with safety guarantees as follows:
	As long as the system state evolution under the performance controller is contained in the
	safe invariant set, the safety mechanism passively monitors proposed control signals.
	However, as soon as the system state would leave the safe set, the safety mechanism actively
	overrules the performance controller and leverages the safety controller to render the
	safe set invariant.
	As the computation of the required safe set and safety controller is very difficult
	in general, the resulting control performance is often limited due to conservative
	approximations or the required design computations are restricted to small-scale
	systems up to 3-4 state dimensions due to the curse of dimensionality.
	\par
	To overcome this limitation, recent concepts extend potentially conservative safe
	sets and safety controllers during closed-loop operation by a just-in-time
	computation of safe backup plans from the current system state, which are required to terminate
	in a potentially conservative safe set.
	These methods are also known as active set reachability~\cite{Gurriet2018},
	SHERPA~\cite{Mannucci2017}, model predictive safety certification (MPSC)~\cite{Wabersich2019},
	predictive safety filters (PSF)~\cite{Wabersich2018a}, and predictive shielding~\cite{li2020robust}.
	While some of them are based on MPC 
	techniques~\cite{Wabersich2018,Wabersich2019,Wabersich2018a,li2020robust},
	other approaches, e.g.~\cite{Gurriet2018}, extend a conservative safe set online through an
	explicit safety controller.
	As shown in~\cite{Gurriet2020}, these approaches can significantly reduce conservativeness
	and thereby increase acceptance of such safety mechanisms.
	\par
	%% Specific limitations and motivation for the paper
	Despite promising theoretical and experimental results, the major drawback
	of such just-in-time safety computations are unexpected external disturbances
	or constraint violating initial conditions, which cannot always be anticipated
	at the controller design stage. In such cases, even for minor errors in design
	assumptions, the corresponding online problem can become infeasible, failing
	to provide a safe input, when it would be most crucial to recover the system
	from constraint violations.
\subsection{Contributions}
	This paper proposes a predictive barrier function approach to recover infeasible
	predictive controllers in an asymptotically stable fashion through a soft
	constraint recovery mechanism.
	In particular, we consider recovery of predictive safety
	filter problems~\cite{Wabersich2018,Wabersich2019,Wabersich2018a}, which
	are commonly used in combination with learning-based control techniques
	to ensure safety and which are lacking intrinsic stability properties compared
	to, e.g., classical MPC~\cite{rawlings2009model}.
	Different from a simple softening of the state constraints `out-of-the-box'
	\cite{Kerrigan2000}, we propose a softened predictive safety filter problem,
	denoted predictive barrier function problem, which renders the feasible set of
	the predictive safety filter Lyapunov stable, thereby enabling a reliable
	recovery of the safety guarantees from unexpected disturbances or unsafe
	initial conditions.
	\par
	The primary mechanism of the predictive barrier function is to employ
	a control barrier function on the last predicted state with a corresponding
	terminal safe set.
	Through a generalization of existing discrete-time control barrier function theory,
	we formally establish that the value function of the predictive barrier function
	problem itself qualifies as a control barrier function, thereby enlarging the
	region of attraction of a potentially conservative terminal control barrier function.
	As a result, we obtain Lyapunov stability properties with respect to the feasible
	safe set of the original predictive safety filter problem.
	\par
	The novel relation between predictive safety filters and control barrier functions
	additionally allows for combining recent results from control barrier function literature,
	such as safe reinforcement learning with convergence guarantees~\cite{Marvi2020}, together
	with significantly less restrictive predictive safety filters.
	Furthermore, since the optimal value function of the predictive barrier function problem
	represents a continuous measure of safety, it can efficiently be approximated using,
	e.g., artificial neural networks. This allows for a practical implementation for
	highly nonlinear dynamical systems, where the corresponding online problem
	to evaluate the predictive barrier function is challenging to solve in real-time.
	\par
	%% Numerical examples
	We provide a design procedure to synthesize the proposed method and demonstrate
	the ability to recover infeasible predictive safety filters	using numerical simulations,
	i.e., a basic linear 2D example and a nonlinear vehicle example.

\subsection{Related work}
	%% Robust & Stochastic predictive safety filters
	While online predictive safety mechanisms that assume a perfect system
	knowledge~\cite{Gurriet2018,Gurriet2020} can quickly become infeasible,
	robust and stochastic
	formulations~\cite{Mannucci2017,Wabersich2018,Wabersich2019,Wabersich2018a,li2020robust}
	provide safety through recursive feasibility of the problem despite external disturbances.
	Nevertheless, the underlying assumption of these techniques is a feasible initial
	condition as well as an accurate uncertainty description of the disturbance, which can be
	difficult to ensure and can cause controller failure if not satisfied.
	We overcome this limitation by proposing a concept for state constraint relaxation that
	provides a feasible problem to compute stabilizing inputs, even for cases, in which
	robust or stochastic formulations would become infeasible.
	\par
	%% Soft constrained MPC approaches
	Infeasibility issues of online predictive control problems have also been
	investigated in the model predictive control (MPC) literature, e.g.~\cite{Kerrigan2000,Zeilinger2014}, the results are, however, limited to linear systems with polytopic state constraints.
	Furthermore, stability properties of the closed-loop system are typically given with respect
	to a steady-state rather than the original feasible set of the hard constrained model predictive
	control problem, which can lead to a longer duration of constraint violation
	in closed-loop in favor of stability with respect to the origin.
	Related to these techniques, the notion of input-to-state stability in
	MPC~\cite{limon2009input} can provide similar recovery properties from disturbances
	or infeasible initial conditions, which are, however, also generally coupled to a
	stability analysis with respect to the origin and rely on a positive definite cost
	function in the MPC problem together with additional assumptions in the nonlinear case.
	\par
	%% Barrier function relaxation (Ebenbauer)
	An alternative MPC-related technique to provide feasibility of online
	MPC problems is based on so-called any-time MPC algorithms~\cite{Wills2004,Feller2017}
	that use a relaxation at the optimization algorithm level, to ensure
	similar stability properties as presented in~\cite{Zeilinger2014}.
	This approach is also tied to a positive definite stage cost function, linear dynamics
	and polytopic constraints, as well as a particular optimization approach
	for solving the MPC problem. By ensuring feasibility through a modified formulation of
	the predictive control problem, the presented approach allows for leveraging recent developments
	in the active field of fast real-time optimization
	techniques~\cite{domahidi2012efficientInteriorPoint,Verschueren2019}
	without relying on a specific optimization algorithm.
	\par
	%% Combining CBF with MPC
	Combining MPC control design with control Lyapunov or control barrier functions
	has also been proposed in~\cite{8887800,8431468,Zeng2020}, where the idea is
	to impose an explicit barrier function constraint on each predicted system state,
	which provides guarantees in terms of constraint satisfaction. While there
	also exists a soft-constrained formulation using control Lyapunov functions~\cite{Grandia2020}, the main
	limitation of these approaches is that each predicted state must be contained in the domain of the
	control barrier/control Lyapunov function. As the design of these functions is a challenging task
	for general systems,~their region of attraction is often restricted to small subsets around linearization points and can significantly limit the set of states for which the resulting MPC provides	theoretical guarantees. In contrast, we impose softened state constraints along the planning horizon and only combine the last predicted state with a control barrier function, which allows for a simple design procedure with an increased feasible set.
\subsection{Notation and common definitions}
	The distance between a vector $x\in\RR^n$ and a set $\AA \subseteq \RR^n$ is
	defined as $|x|_{\AA}\mDef\inf_{y\in\AA} ||x-y||$. By $\mOnes{n}$ we denote
	the vector of ones with length $n\in\NN^+$ and $\max(x,y)$ with $x,y\in\RR^n$
	equals the element-wise maximum value. The set of
	symmetric positive definite matrices of size $n\times n$ is denoted
	by $\mSetPosSymMat{n}$.
	\begin{definition}\label{def:forward_invariance}
		A set $\AA \subseteq \RR$ is said to be positively
		invariant for the autonomous system~$x(k+1) = g(x(k))$
		if for every $x(0)\in\AA$ it holds for all $k\in\NN$ that
		$x(k)\in\AA$.
	\end{definition}
\section{Preliminaries}\label{sec:psf}
	\begin{figure}
		\centering
		\begin{tikzpicture}[scale=0.8]
			\input{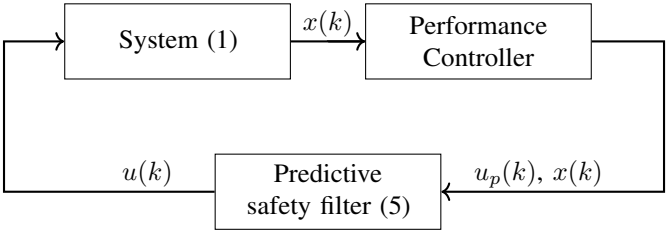}
		\end{tikzpicture}
		\caption{The concept of a model predictive safety filter: Using the current state $x(k)$,
			a task-specific performance controller provides an input $u_p(x(k))$, which is processed
			by the predictive safety filter~\eqref{eq:opt_PSF} and then applied to the real
			system~\eqref{eq:system}.}
		\label{fig:concept}
	\end{figure}
	\begin{figure}
		\centering
		\begin{tikzpicture}[scale=0.9]
			\input{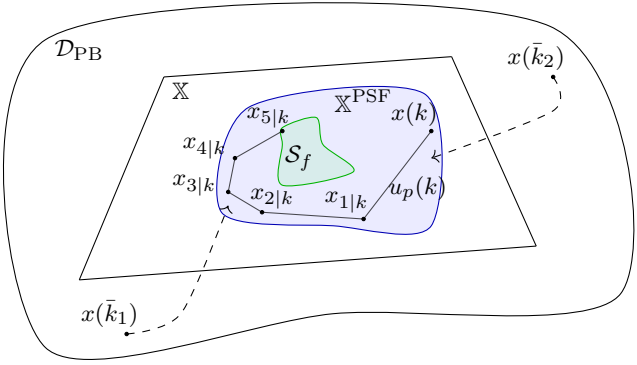}
		\end{tikzpicture}
		\caption{Inner workings of the predictive safety filter~\eqref{eq:opt_PSF}:
		Based on the current state $x(k)$, the predicted state sequence terminates in the terminal safe
		set $\SS_f$ (green set) after 5 planning steps. The feasible set of the safety filter
		problem $\XX^{\mathrm{PSF}}$~\eqref{eq:feasible_set_psf}, i.e. the set of states from which
		$\SS_f$ can be reached within $5$ time steps, is illustrated in blue and is a subset of
		the state constraint set $\XX$ (black lines). For system states $x(\bar k_1)$ and $x(\bar k_2)$
		outside the feasible set $\XX^{\mathrm{PSF}}$, the safety filter cannot provide
		a safe input. We address this limitation in this work by introducing a softened reformulation
		that stabilizes	the feasible set $\XX^{\mathrm{PSF}}$ with enlarged region of attraction
		$\DD_{\mathrm{PB}}$ in Section~\ref{sec:predictive_cbf}.}
		\label{fig:psf}
	\end{figure}
	%
	%% System class and constraints
	We consider discrete-time control systems of the form
	\begin{align}\label{eq:system}
			x(k+1) = f(x(k),u(k)),~k\in\NN,
	\end{align}
	with continuous dynamics $\mDefFunction{f}{\RR^n \times \RR^m}{\RR^n}$
	and initial condition $x(0)=x_0$ with $x_0\in\RR^n$.
	The system is subject to hard physical input constraints of the form
	\begin{align}\label{eq:input_constraints}
		u(k) \in \UU \subset \RR^m,
	\end{align}
	with compact $\UU$, and is required to satisfy
	safety constraints in the form of state constraints given by
	\begin{align}\label{eq:state_constraints}
		x(k) \in \XX \mDef \{x\in\RR^n: c(x)\leq 0\} \subset \RR^n,
	\end{align}
	for all time steps $k\in\NN$ where $\mDefFunction{c}{\RR^n}{\RR^{n_x}}$
	is a continuous function.
	While the input constraints~\eqref{eq:input_constraints}
	can be generic compact subsets of $\RR^m$, the state constraints~\eqref{eq:state_constraints}
	need to be represented as a set of nonlinear inequalities, which include,
	e.g., box, ellipsoidal, or polytopic constraints as relevant special cases.
	Relevant learning-based control applications or complex tasks via RL,
	for which a sufficiently accurate model~\eqref{eq:system} is available
	include inverse optimal control, imitation learning, or interactions with
	humans~\cite[Section 4]{Hewing2020}.
	In the following, we briefly review the predictive safety filter formulation 
	as shown in Figure~\ref{fig:concept} and introduce a soft-constrained extension.
	\subsection{Predictive Safety Filter}\label{subsec:psf}
	% PSF objective
	The overall goal is to provide a modular approach to safety of task-specific
	performance controllers
	\begin{align}\label{eq:performance_control}
		\mDefFunction{u_p}{\RR^n}{\RR^m}
	\end{align}
	as originally proposed in~\cite{Seto1998}, see Figure~\ref{fig:concept}.
	Potentially unsafe inputs $u_p(x(k))$ are processed in real-time by a safety
	filter that decides based on the current system state if the input is safe
	to apply or if it needs to be modified.
	To implement the desired safety filter, we make use of an
	MPC-based concept called predictive safety filter as introduced
	in~\cite{Wabersich2018}.
	In particular, our goal is to extend the MPC-based
	mechanism from~\cite{Wabersich2018} to not only keep the system
	safe but to also recover the safety filter from constraint violations, which can
	result from unexpected external disturbances or infeasible initial conditions.
	To this end, we first recap the nominal formulation of a predictive
	safety filter as introduced in~\cite[Section~5]{Hewing2020} in the following,
	which we will then equip with a recovery mechanism in Section~\ref{sec:predictive_cbf}.
	\par
	% Plain PSF
	A basic model predictive safety filter, realizing the safety filter
	block in Figure~\ref{fig:concept}, is given by
	\begin{subequations}\label{eq:opt_PSF}
	\begin{align}
		\min_{u_{i|k}}~~& \mNormGenSmall{u_p(x(k)) - \mUpred_{0|k}}\label{eq:opt_PSF_cost}\\
		\text{s.t.}  	~~& \text{for all } i=0,..,N-1: \nonumber\\
						~~& \mXpred_{0|k} = x(k), \label{eq:opt_PSF_init}\\
						~~& \mXpred_{i+1|k} = f(\mXpred_{i|k},\mUpred_{i|k}), \label{eq:opt_PSF_dybamics}\\
						~~& \mUpred_{i|k} \in \UU, \label{eq:opt_PSF_input_cons}\\
						~~& \mXpred_{i|k} \in \XX, \label{eq:opt_PSF_state_cons}\\
						~~& \mXpred_{N|k}\in\SS_f, \label{eq:opt_PSF_terminal}
	\end{align}
	\end{subequations}
	where the resulting input applied to the system is given by $u(k) = \mUpredOpt_{0|k}$.
	We use the subscript $i|k$ to emphasize predictive quantities, where, e.g., $x_{i|k}$ is
	the $i$-step-ahead prediction of the state, initialized at $x_{0|k} = x$ at time step $k$.
	Optimal states and inputs will be denoted by an asterisk, e.g. $\mUpredOpt_{i|k}$
	or $\mXpredOpt_{i|k}$.
	The objective~\eqref{eq:opt_PSF_cost} is to achieve minimal deviation between the
	first element of the predicted input sequence, $\mUpredOpt_{0|k}$, and the currently
	requested performance input, $u_p(x(k))$, while satisfying state~\eqref{eq:opt_PSF_state_cons}
	and input constraints~\eqref{eq:opt_PSF_input_cons} for predicted time steps.
	By solving~\eqref{eq:opt_PSF} at every time step, we obtain the desired filtering
	property as follows. If for a specific state $x(k)$ and proposed performance input
	$u_p(x(k))$ the objective~\eqref{eq:opt_PSF_cost} is zero subject to all relevant system
	constraints for all future time steps, then the performance controller is safe and
	will directly be applied since $\mUpredOpt_{0|k}=u_p(x(k))$.
	However, if the objective is greater than zero, the performance control input $u_p(x(k))$
	at the current time step cannot be verified as safe and the safety filter mechanism
	overwrites the performance controller, i.e. $\mUpredOpt_{0|k}\neq u_p(x(k))$.
	The PSF problem makes use of a terminal constraint on the
	last predicted state~\eqref{eq:opt_PSF_terminal}, see Figure~\ref{fig:psf} (green set),
	which is typically assumed to be positive invariant for system~\eqref{eq:system} (Definition~\ref{def:forward_invariance}) under a local safety control law satisfying state and input constraints.
	This technique from MPC literature ensures that initial feasibility of~\eqref{eq:opt_PSF}
	at state $x(k)$ with corresponding optimal input sequence $\mUpredOpt_{i|k}$ implies
	feasibility of~\eqref{eq:opt_PSF} for future time steps in a recursive fashion.
	This can be shown easily through constructing a feasible candidate solution
	for time step $k+1$, which is given by
	$\mUpredCand_{i|k+1}=\{\mUpredOpt_{1|k},\mUpredOpt_{2|k},..,\mUpredOpt_{N|k},\tilde u_f\}$,
	where $\tilde u_f$ is selected such that $\mXpredCand_{N|k+1}=f(\mXpredOpt_{N|k},\tilde u_f)\in\SS_f$
	holds, which is possible due to the invariance property of $\SS_f$~\cite{rawlings2009model,
	chen1998quasiInfiniteHorizonMPC}.
	\par
	%% Recursive feasibility and constraint satisfaction
	Recursive feasibility of~\eqref{eq:opt_PSF} directly implies constraint satisfaction
	for all future time steps together with the fact that the feasible set
	\begin{align}\label{eq:feasible_set_psf}
		\XX^{\mathrm{PSF}}\mDef\{x\in\RR^n ~|~\eqref{eq:opt_PSF}\text{ is feasible}\},
	\end{align}
	as illustrated in Figure~\ref{fig:psf} (blue set), is rendered forward invariant
	under $u(k)=\mUpredOpt_{0|k}$ and realizes an implicitly defined safe set.
	\par
	While there exist extensions of~\eqref{eq:opt_PSF} that can deal with external disturbances in a
	robust~\cite{Wabersich2018a} or stochastic manner~\cite{Wabersich2019}, these methods rely
	on specific assumptions on the disturbance in the form of a known magnitude or probability
	distribution.
	These techniques therefore still suffer from infeasibility issues in case of
	unmodeled external disturbances $d(k)\in\RR^n$ such that
	$f(x(k),u(k))+d(k)\notin\XX^{\mathrm{PSF}}$ or even $f(x(k),u(k))+d(k)\notin\XX$.
	These situations commonly arise in practice with severe consequences, as no sensible
	input can be generated from an infeasible PSF problem~\eqref{eq:opt_PSF} to ensure system safety,
	as illustrated with states $x(\bar k_1)$ and $x(\bar k_2)$ in Figure~\ref{fig:psf}.
	The goal in the following is to render the feasible set of the safety filter~\eqref{eq:feasible_set_psf}
	in such situations asymptotically stable through a soft-constrained recovery mechanism.
	\subsection{Simplistic soft constrained predictive safety filter}\label{subsec:soft_constraints}
	\begin{figure}
		\centering
		\begin{tikzpicture}[scale=0.8]
			\input{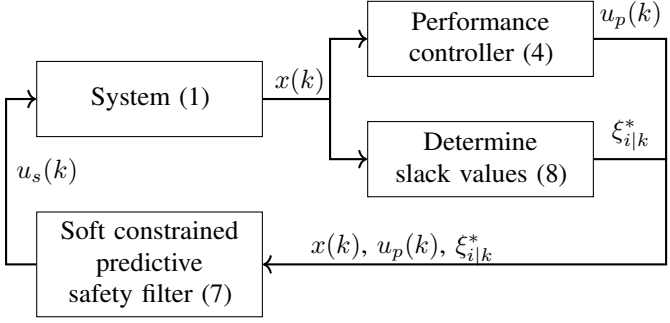}
		\end{tikzpicture}
		\caption{Predictive safety filter with a simplistic soft constrained recovery mechanism:
			Based on the current state $x(k)$, the performance controller~\eqref{eq:performance_control}
			provides a desired control input $u_p(x(k))$. At the same time,
			slack values~$\xi_{i|k}^*$ are computed by either solving the simple
			feasibility problem~\eqref{eq:opt_pre_soft_PSF} according to Section~\ref{subsec:psf}
			or by evaluating the predictive barrier function~\eqref{eq:opt_BF} as described in
			Section~\ref{sec:predictive_cbf}, both ensuring existence of a feasible solution
			to the softened predictive safety filter problem~\eqref{eq:opt_soft_PSF}
			for all $x(k)\in\RR^n$.}
		\label{fig:psf_soft_concept}
	\end{figure}
	%
	%% Softening for nonlinear MPC
	In practical implementations of the predictive safety filter, a softening of the state
	constraints as proposed, e.g., in \cite{DiPillo1989,Kerrigan2000} for MPC, can guarantee
	feasibility of the online optimization problem~\eqref{eq:opt_PSF} despite unexpected disturbances.
	In this section, we first recap a simplistic `separation of objectives' approach and discuss
	its limitations, which motivates the predictive control barrier function presented in
	Section~\ref{sec:predictive_cbf}. The simplistic soft constrained PSF problem is given by
	\begin{subequations}\label{eq:opt_soft_PSF}
	\begin{align}
		\min_{u_{i|k}}~~& \mNormGenSmall{u_p(x(k)) - \mUpred_{0|k}}\label{eq:opt_soft_PSF_cost}\\
		\text{s.t.}  	~~& \text{for all } i=0,..,N-1: \nonumber\\
						~~& \eqref{eq:opt_PSF_init}-\eqref{eq:opt_PSF_input_cons}, \label{eq:opt_soft_PSF_hard_cons}\\
						~~& \mXpred_{i|k} \in \XX(\xi_i^*), \label{eq:opt_soft_PSF_state_cons}\\
						~~& \mXpred_{N|k}\in\SS_f(\xi_N^*). \label{eq:opt_soft_PSF_terminal}
	\end{align}
	\end{subequations}
	Thereby, \eqref{eq:opt_soft_PSF_state_cons}, \eqref{eq:opt_soft_PSF_terminal} are soft
	constraints of the form $\XX(\xi_{i|k}^*)=\{x\in\RR^n|c(x)\leq \xi_{i|k}^*\}$ with pre-computed
	slacks according to
	\begin{subequations}\label{eq:opt_pre_soft_PSF}
	\begin{align}
		\xi_{i|k}^* = \argmin_{\xi_{i|k},u_{i|k}}~~&~\sum_{i=0}^{N} \Vert \xi_i \Vert
		\label{eq:opt_pre_soft_PSF_obj}\\
		\text{s.t.}  	~~& \text{for all } i=0,..,N-1:\\
						~~& \eqref{eq:opt_PSF_init}-\eqref{eq:opt_PSF_input_cons} \\
						~~& \mXpred_{i|k} \in \XX(\xi_i),~~0\leq \xi_i,\\
						~~& \mXpred_{N|k}\in\SS_f(\xi_N),~~0\leq \xi_N,
	\end{align}
	\end{subequations}
	as depicted	in Figure~\ref{fig:psf_soft_concept}. This ensures feasibility of~\eqref{eq:opt_soft_PSF} for any input sequence,
	even if $\mXpred_{i|k}\notin\XX$, in which case the violation is penalized in the objective for determining the slack\
	values~\eqref{eq:opt_pre_soft_PSF_obj}.
	By construction, this two-step approach provides the same optimal solution
	as~\eqref{eq:opt_PSF} for all $x\in\XX^{\mathrm{PSF}}$ and preserves feasibility otherwise.
	\begin{remark}
		Although the additional feasibility problem~\eqref{eq:opt_pre_soft_PSF} introduces	a second optimization
		stage at each sampling time, it is important to note that it can be solved in parallel
		with the performance controller as it only depends on the current system state.
		Additionally,~\eqref{eq:opt_pre_soft_PSF} provides a feasible warm start of~\eqref{eq:opt_soft_PSF},
		that can be leveraged by many optimization algorithms.
	\end{remark}                                                              
	While \eqref{eq:opt_soft_PSF} and \eqref{eq:opt_pre_soft_PSF} ensure feasibility, it is
	unclear if the resulting control input drives the state $x(k)$ back into $\XX^{\mathrm{PSF}}$ at
	some future point in time from $x(0)\notin\XX^{\mathrm{PSF}}$. In the following, we address this issue by
	introducing a modified formulation of the feasibility problem~\eqref{eq:opt_pre_soft_PSF},
	called predictive barrier function problem, that ensures feasibility of the soft constrained predictive
	safety filter problem~\eqref{eq:opt_soft_PSF} and is guaranteed to recover feasibility of \eqref{eq:opt_PSF}. 
	\section{Predictive control barrier functions}\label{sec:predictive_cbf}
	%% State the presolve problem and discuss relation to exact penalty/forward invariance	
	The proposed predictive control barrier function problem reads
	\begin{subequations}\label{eq:opt_BF}
	\begin{align}
		\mH(x)
			\mDef\min_{u_{i|k},\xi_{i|k}}~~&\alpha_f \xi_{N|k}
			+ \sum_{i=0}^{N-1}\mNormGenSmall{\xi_{i|k}} \label{eq:opt_BF_cost}\\
		\text{s.t.}  	~~& \text{for all } i=0,..,N-1: \nonumber\\
						~~& \mXpred_{0|k} = x, \label{eq:opt_BF_init}\\
						~~& \mXpred_{i+1|k} = f(\mXpred_{i|k},\mUpred_{i|k}), \label{eq:opt_BF_dynamics}\\
						~~& \mXpred_{i|k} \in \XX_i(\xi_{i|k}), ~ 0\leq \xi_{i|k},\label{eq:opt_BF_state_cons}\\
						~~& \mUpred_{i|k} \in \UU, \label{eq:opt_BF_input_cons}\\
						~~& h_f(\mXpred_{N|k})\leq\xi_{N|k}, ~0\leq\xi_{N|k}, \label{eq:opt_BF_terminal}
	\end{align}
	\end{subequations}
	which implements a modified (tightened) state constraint~\eqref{eq:opt_BF_state_cons}
	as detailed in Section~\ref{subsec:tightend_constraints} and a more specific
	terminal constraint~\eqref{eq:opt_BF_terminal} as formalized in
	Section~\ref{subsec:terminal_configuration} compared to~\eqref{eq:opt_pre_soft_PSF}.
	Problem~\eqref{eq:opt_BF} thereby provides a Lyapunov-like value function
	$\mH(x)\geq 0$,~which ensures asymptotic stability of the set
	of states
	\begin{align}\label{eq:safe_set_BF}
		\SS_{\mathrm{PB}}=\{x\in\RR^n:\mH(x)=0\},
	\end{align}
	for which we can find an optimal solution to~\eqref{eq:opt_BF} without state constraint violations
	along the prediction. To this end, we guarantee a decrease of $\mH(x)$ for all
	$x\notin\SS_{\mathrm{PB}}$ between two consecutive time-steps under application of
	$u(k)=\mUpredOpt_{0|k}$.
	\par
	The decrease is achieved via two components. An iterative tightening of the state constraints
	along predictions in~\eqref{eq:opt_BF_state_cons}, which is detailed in
	Section~\ref{subsec:tightend_constraints}, ensures a reduction of the slacks
	$\sum_{i=0}^{N-1}\mNormGenSmall{\xi^*_{i|k}}$ from one time step to the next if $\xi^*_{i|k}\neq 0$ for
	some $i=1,..,N-1$.
	A decrease of the term $\alpha_f\xi^*_{N|k}$ in~\eqref{eq:opt_BF_cost}
	is obtained by introducing a terminal set of the form $\SS_f\mDef \{x\in\RR^n|h_f(x)\leq 0\}$ and
	requiring $\SS_f$ to be a safe set corresponding to a so-called control barrier function $h_f(x)$ as
	it will be introduced in Section~\ref{subsec:terminal_configuration}.
	\par
	These modifications will allow us in Section~\ref{subsec:theoretical_guarantees} to
	establish that the proposed scheme as summarized in Algorithm~1 renders $\SS_{\mathrm{PB}}$ in~\eqref{eq:safe_set_BF} asymptotically stable.
	Note that due to the constraint tightening in~\eqref{eq:opt_BF_state_cons}, the set $\SS_{\mathrm{PB}}$	is a subset of the nominal feasible set of the PSF problem~\eqref{eq:feasible_set_psf}, i.e.
	$\SS_{\mathrm{PB}}\subset \XX^{\mathrm{PSF}}$. Establishing asymptotic stability 
	of $\SS_{\mathrm{PB}}$ therefore implies the desired recovery from constraint violations.~
	%%
	%% Algorithm
	\begin{algorithm}[t]
		\caption{Predictive safety filter with recovery mechanism}
		\label{alg:half_space_prs}
		\begin{algorithmic}[1]
			% \REQUIRE Performance controller $u_p(.)$
			\FOR{$k=0,1,2..$}
				\STATE Measure system state $x(k)$
				\STATE Evaluate desired performance input $u_p(k)$ and determine slack
					values $\{\xi^*_{i|k}\}$ by solving~\eqref{eq:opt_BF} in parallel
				\STATE Solve
					\begin{align*}
					\{\mUpredOpt_{i|k}\} \in \argmin_{u_{i|k}}~~& \mNormGenSmall{u_p(x(k)) - \mUpred_{0|k}}\\
					\text{s.t.}  	~~& \text{for all } i=0,..,N-1: \nonumber\\
									~~& \eqref{eq:opt_PSF_init}-\eqref{eq:opt_PSF_input_cons}, \\
									~~& \mXpred_{i|k} \in \XX_i(\xi^*_{i|k}), \\
									~~& h_f(\mXpred_{N|k})\leq\xi^*_{N|k}.
					\end{align*}
				\STATE Apply $u(k)\leftarrow \mUpredOpt_{0|k}$ to system~\eqref{eq:system}
			\ENDFOR
		\end{algorithmic}
	\end{algorithm}
	%%%
	\subsection{Tightened soft constraints}\label{subsec:tightend_constraints}
	While the simple state and terminal constraint softening from Section~\ref{subsec:soft_constraints}
	ensures feasibility, e.g., if $\xi^*_{0|k}=0$, $\xi^*_{j|k}\neq 0$, and $\xi^*_{N|k}=0$ for some
	$j>0$, it is not clear, whether the total amount of slack $\sum_{i=0}^{N} \Vert \xi_{i|k}^* \Vert$ will be reduced
	at the next time step $k+1$ using the scheme as depicted in Figure~\ref{fig:psf_soft_concept}. More precisely,
	the construction of a feasible candidate at time $k+1$ using a common shifting operation of the optimal
	solution at time $k$ would only keep the amount of total slack constant.~
	This is due to the fact that the resulting stage cost function, given by~$\|\xi_{i|k}^*(x(k))\|=\|\max(0, c(x_{i|k}^*(x(k))))\|$,
	is not positive definite in $x(k)$ w.r.t. the desired target set $\SS_{\mathrm{PB}}$~\eqref{eq:feasible_set_psf} and therefore prohibits
	application of existing MPC stability theory~\cite{rawlings2009model}.
	\par
	To guarantee an overall decrease of $\sum_{i=0}^{N}\|\xi_{i|k}\|$ towards zero in this case,
	we perform a simple state constraint tightening along the prediction horizon. This
	tightening is required to render the constraints more restrictive with each
	prediction step. Given a successfully solved instance of~\eqref{eq:opt_BF},
	this allows to construct a candidate solution to~\eqref{eq:opt_BF} at the next time-step,
	with a reduced slack sequence according to the tightening increments between two prediction steps.
	% Details can be found in the last part of the proof of Theorem~\ref{thm:BF}.
	%
	\par
	The iterative tightening~\eqref{eq:opt_BF_state_cons} of the state
	constraints~\eqref{eq:state_constraints} along the planning horizon $i=0,1,..,N-1$
	is defined as
	\begin{align}
		\XX_{i}\mDef\{x\in\RR^n:c_j(x) \leq -\Delta_i~\forall j\}
	\end{align}
	for a strictly increasing sequence $\Delta_i$ with $\Delta_0=0$.
	The overall corresponding soft constraints are defined as
	\begin{align}\label{eq:def_tight_state_constraints}
		\XX_{i}(\xi)\mDef\{x\in\RR^n:c(x) \leq -\Delta_i\mOnes{} + \xi\}
	\end{align}
	with $\xi\geq 0$~according to~\eqref{eq:opt_BF_state_cons}.
	\par
	While the resulting stage cost in~\eqref{eq:opt_BF_cost} with $\xi_{i|k}^*=\max(0, c(x_{i|k}^*)+\Delta_i)$
	is still not positive definite w.r.t. the set~\eqref{eq:feasible_set_psf}, we can
	construct a modified feasible candidate slack sequence
		$\bar \xi_{i|k+1} = \xi_{i+1|k} +(\Delta_{i} - \Delta_{i+1})\mOnes{}$
	at time $k+1$ due to the tightening~\eqref{eq:def_tight_state_constraints}, which ensures a monotonic
	decrease of $\sum_{i=0}^N \|\bar \xi_{i|k+1}\|$ even in the case
	that $\xi_{N|k}=0$ as discussed earlier in this section. This will play a central role in Section~\ref{subsec:theoretical_guarantees}
	to establish asymptotic stability of the feasible set $\SS_{\mathrm{PB}}$~\eqref{eq:feasible_set_psf}.
	Mechanisms to ensure a decrease even if we cannot reach the terminal safe set $\SS_f$ within the
	prediction horizon, i.e. $\bar\xi_{N|k+1}>0$, will be discussed next.~
	\subsection{Terminal control barrier function}\label{subsec:terminal_configuration}
	To ensure a decrease towards $\sum_{i=0}^N \|\xi_{i|k}\|=0$ in cases where
	$\xi_{N|k}\neq 0$, we design the terminal constraint by selecting $h_f$ to
	be a so-called control barrier function with a corresponding terminal safe set $\SS_f$. 
	\begin{definition}\label{def:barrier_function}
		Function $\mDefFunction{h}{\DD}{\RR}$ is called a discrete-time
		control barrier function with a corresponding safe set
		$\SS\mDef\{x\in\RR^n:h(x)\leq 0\}\subset\DD$,
		if $\SS$ and $\DD$ are non-empty and compact, $h(x)$ is continuous on $\DD$, and
		if there exists a continuous function $\mDefFunction{\Delta h}{\DD}{\RR}$
		with $\Delta h(x) >0$ for all $x \in \DD\setminus \SS$ such that 
		\begin{subequations}\label{eq:def_barrier_function_12}
			\begin{align}
				\forall x\in\DD\setminus\SS: & \inf_{u\in\UU} h(f(x,u))-h(x) \leq -\Delta h(x)
					\text{ and}
					\label{eq:def_barrier_function_1}\\
				\forall x\in\SS:&\inf_{u\in\UU} h(f(x,u))\leq 0.
					\label{eq:def_barrier_function_2}
			\end{align}
		\end{subequations}
		The set of safe control inputs at $x\in\DD$ w.r.t. $h$ is given by
		\begin{align}
			K_{\mathrm{CBF}}(x)\mDef
			\begin{cases}
				K_{\mathrm{CBF}}^1(x),~x\in\DD\setminus\SS,
				\\
				K_{\mathrm{CBF}}^2(x),~x\in\SS
			\end{cases}
			\label{eq:def_barrier_function_3}
		\end{align}
		with $K_{\mathrm{CBF}}^1(x)\mDef\{u\in\UU:h(f(x,u))-h(x)\leq-\Delta h(x)\}$
		and $K_{\mathrm{CBF}}^2(x)\mDef\{u\in\UU:h(f(x,u))\leq 0\}$.\END
	\end{definition}
	\begin{assumption}\label{ass:terminal_cbf}
		The function $\mDefFunction{h_f}{\DD_f}{\RR}$
		in~\eqref{eq:opt_BF_terminal} is a control barrier function
		according to Definition~\ref{def:barrier_function}
		with corresponding safe set denoted by~$\SS_f$.
	\end{assumption}
	\begin{remark}\label{rem:definition_control_barrier_function}
		While Definition~\ref{def:barrier_function} is inspired by common control barrier function
		concepts~\cite{Ames2019,Ohnishi2019} (see, e.g., \cite[Remark 3]{Ames2019} for
		alternative formulations), we do not require an exponential decrease
		in~\eqref{eq:def_barrier_function_1}, i.e.~$-\Delta h(x)=-\gamma(h(x))$ for some
		extended $\KK_\infty$ function $\gamma$. Instead, we only require the existence of a
		continuous function $\Delta h(x)$, bounding the decrease between two
		consecutive time steps.
		~Available discrete-time control barrier function design techniques can therefore be applied
		to satisfy Assumption~\ref{ass:terminal_cbf}. In addition, we present a principled
		design procedure in Section~\ref{sec:terminal_cbf_design} for linear and nonlinear
		systems with polytopic constraints using model predictive control
		related techniques.
	\end{remark}
	From Assumption~\ref{ass:terminal_cbf}, we can directly conclude that in case
	$\mXpredOpt_{N|k}\notin \SS_f$ there exists a candidate input and slack $\bar u_{N-1|k+1}$
	and $\bar \xi_{N|k+1}$ that yield a negative cost difference $\alpha (\bar \xi_{N|k+1} - \xi^*_{N|k})\leq
	- \Delta h_f(\mXpredOpt_{N|k})<0$, which can be scaled through the parameter $\alpha >0$ in~\eqref{eq:opt_BF_cost}.
	Together with the constraint tightening	from Section~\ref{subsec:tightend_constraints}, this
	will allow us to establish	the desired asymptotic stability properties in
	the following.
	\subsection{Theoretical analysis}\label{subsec:theoretical_guarantees}
	In this section, we formally show asymptotic stability of the feasible safe set
	of states $S_{\mathrm{PB}}$ in~\eqref{eq:safe_set_BF} under application of Algorithm~1.
	% \par
	To this end, we first recap Lyapunov stability of invariant
	sets similar to~\cite[Appendix B.2]{rawlings2009model}.
	\par
	In a second step, we establish an intermediate result, implying that the smaller
	terminal safe set $\SS_f\subset\SS_{\mathrm{PB}}$ according to
	Assumption~\ref{ass:terminal_cbf} can be rendered asymptotically stable
	within an enlarged terminal domain $\DD_f$, $\SS_f \subset \DD_f$ by applying
	safe terminal control inputs $u(k)\in K_{\mathrm{CBF},f}(x)$. This is done
	by relating the control barrier function $h_f(x)$ to the Lyapunov
	stability results from the first step.
	\par
	We then combine these results together with the constraint tightening
	mechanism from Section~\ref{subsec:tightend_constraints} in a third step
	to prove that the predictive barrier function $\mH(x)$ in~\eqref{eq:opt_BF}
	is a control barrier function according to Definition~\ref{def:barrier_function}
	with desired safe set $\SS_{\mathrm{PB}}$, which will be rendered asymptotically stable
	under application of Algorithm~1. To this end, we define
	a corresponding domain $\DD_{\mathrm{PB}}\supset \SS_{\mathrm{PB}}$, for which we
	establish local continuity of $\mH(x)$ as well as the decrease~\eqref{eq:def_barrier_function_12}
	between consecutive	time steps. Figure~\ref{fig:set_visualization_analysis} visualizes
	the relations between the different sets.
	\\
	\subsubsection{Lyapunov stability with respect to invariant sets}\label{par:lyapunov_stability}
	Consider the autonomous closed-loop system~\eqref{eq:system}
	\begin{align}\label{eq:autonomous_system}
		x(k+1) = f(x(k),\kappa(x(k))) =: g(x(k)),~k \in \NN,
	\end{align}
	under application of some control law $\kappa(x)$, subject to input and state
	constraints \eqref{eq:input_constraints} and \eqref{eq:state_constraints},
	and with initial condition $x(0)=x_0$. We formalize stability of an
	invariant set with respect to the system~\eqref{eq:autonomous_system} as follows.
	\begin{definition}\label{def:asy_stable}
		Let $\SS$ and $\DD$ be non-empty, compact, and positively
		invariant sets for system~\eqref{eq:autonomous_system}
		such that $\SS \subset \DD$. The set $\SS$ is called an
		asymptotically stable set for system~\eqref{eq:autonomous_system}
		in $\DD$, if for all $x(0)\in\DD$ the following conditions hold:
		\begin{subequations}\label{eq:def_asy_stable}
			\begin{align}
				\forall \epsilon>0~\exists \delta >0:~&
					|x(0)|_\SS<\delta \Rightarrow \forall k > 0:|x(k)|_\SS < \epsilon,
					\label{eq:def_asy_stable_1}\\
				\exists \bar \delta:~&
					|x(0)|_\SS<\bar\delta\Rightarrow\lim_{k\rightarrow\infty}x(k)\in \SS.
					\label{eq:def_asy_stable_2}
			\end{align}
		\end{subequations}
		\END
	\end{definition}
	As shown in Appendix~\ref{app:lyap_proof}, we can extend existing
	continuous-time Lyapunov stability proofs with respect to equilibrium points
	to show the more general case in Definition~\ref{def:asy_stable}.
	This allows us in the next step to show asymptotic stability of safe sets
	according to Definition~\ref{def:barrier_function} under safe
	control inputs by constructing a Lyapunov function from the corresponding
	control barrier function.
	While there are similar results in model predictive control theory
	available \cite[Theorem B.13]{rawlings2009model}, it is important to note that the corresponding assumptions typically rely on a positive defininte stage cost function with respect to $\SS$, which is not present in our case.
	~
	\\
	\subsubsection{Asymptotic stability of the safe terminal set}
	\begin{theorem}\label{thm:cbf}
		
		Let	$\DD\subset \RR^n$ be a non-empty and compact set.
		Consider a control barrier function $\mDefFunction{h}{\DD}{\RR}$
		with $\SS=\{x\in\RR^n:h(x)\leq 0\}\subset\DD$ according to
		Definition~\ref{def:barrier_function}. If $\DD$ is a forward invariant
		set for system~\eqref{eq:system} under $u(k)=\kappa(x(k))$ for all
		$\mDefFunction{\kappa}{\DD}{\UU}$ with $\kappa(x)\in K_{\mathrm{CBF}}(x)$,
		then it holds that
		
		\begin{enumerate}
			\item \label{item:thm_cbf_1}$\SS$ is a forward invariant set,
			\item \label{item:thm_cbf_2}$\SS$ is asymptotically stable in $\DD$.
		\end{enumerate}
	\end{theorem}
	\begin{proof}
		The proof can be found in Appendix~\ref{app:cbf_proof} and is
		based on establishing a relation between control barrier functions
		according to Definition~\ref{def:barrier_function} and Lyapunov stability
		results with respect to sets as presented in Appendix~\ref{app:lyap_proof}.
	\end{proof}
	Note that the set $\DD$ is sometimes also referred to as region of attraction.
	While Theorem~\ref{thm:cbf} together with Assumption~\ref{ass:terminal_cbf} implies invariance of the terminal set $\SS_f$ under application of Algorithm~1, it additionally provides
	asymptotic stability for a superset $\DD_f\supset\SS_f$, which accounts for the case
	that $\SS_f$ cannot be reached within $N$-time steps, i.e. if we encounter non-zero
	terminal slack $\xi_{N|k}^*>0$, see Figure~\ref{fig:set_visualization_analysis}.
	\\
	\subsubsection{The function $\mH(x)$ is a control barrier function}
	The combination of Theorem~\ref{thm:cbf} with the modified soft constraints from
	Sections~\ref{subsec:tightend_constraints} and \ref{subsec:terminal_configuration} finally allows
	us to establish that the optimal value function $\mH(x)$ in~\eqref{eq:opt_BF} is a control
	barrier function with corresponding safe set given by
	\begin{align}\label{eq:predictive_barrier_function_safe_set}
		\SS_{\mathrm{PB}}\mDef\{x\in\RR^n|\mH(x) = 0\}.
	\end{align}
	\par
	While problem~\eqref{eq:opt_BF} is feasible for all $x\in\RR^n$, it is important to note that the
	decrease between consecutive time steps must be
	bounded by a continuous	function according to Definition~\ref{def:barrier_function}.
	Continuity of the dynamics~\eqref{eq:system} and constraints~\eqref{eq:state_constraints}
	can be used to establish a continuous bound on the decrease of~$\sum_{i=0}^{N-1}\mNormGenSmall{\xi^*_{i|k}}$
	as discussed in Section~\ref{subsec:tightend_constraints}. To guarantee a decrease of the term $\alpha_f \xi^*_{N|k}$
	as discussed in Section~\ref{subsec:terminal_configuration}, we need to ensure that $x^*_{N|k}\in\DD_f$ due
	to~\eqref{eq:opt_BF_terminal}. We therefore define the domain $\DD_{\mathrm{PB}}$,
	of $\mH(x)$, using level set concepts from MPC theory~\cite[Chapter 2.6]{rawlings2009model}
	as
	\begin{align}
		\DD_{\mathrm{PB}} \mDef \{x\in\RR^n|h_{\mathrm{PB}}(x)\leq \alpha_f\gamma_f\}.
		\label{eq:domain_BF}
	\end{align}
	with $\gamma_f>0$ such that for all $x\in\RR^n$ with $h_f(x)\leq \gamma_f$ it holds $x\in\DD_f$.
	The relations of the domains $\DD_f$ and $\DD_{\mathrm{PB}}$ are illustrated in
	Figure~\ref{fig:set_visualization_analysis}.
 	\begin{remark}
		The level $\gamma_f$ can be computed for general domains $\DD_f$ by starting
		with $\gamma_f = \max_{x\in\DD_f} h_f(x)$ and iteratively shrinking $\gamma_f$ until
		$\{x|h_f(x)\leq \gamma_f,x\notin\DD_f\}$ is an empty set, which can be verified
		using constrained optimization techniques.~ In
		Section~\ref{sec:terminal_cbf_design} we provide a specific procedure
		to obtain the terminal ingredients $h_f$, $\SS_f$, and $\DD_f$.
	\end{remark}
	Large values of the terminal weight $\alpha_f>0$ in the
	domain $\DD_{\mathrm{PB}}$~\eqref{eq:domain_BF} of $\mH$ support large
	amounts of state constraint violations that we are able to recover.
	More precisely, if we can find a predicted trajectory with terminal
	cost $h_f(\mXpred_{N|k}) \leq \alpha_f\gamma_f$, then the magnitude
	of admissible cumulative state constraint violations
	$\sum_{i=0}^{N-1}\Vert \xi_{i|k}\Vert \leq \alpha_f\gamma_f - h_f(\mXpred_{N|k})$,
	is proportional to $\alpha_f$. From this we can also conclude that the domain \eqref{eq:domain_BF}
	is guaranteed to be larger or equal compared to the domain $\mathcal D_f$ of the
	terminal control barrier function, see Figure~\ref{fig:set_visualization_analysis}.~
	While we formalize in the following the main result that
	a lower bound on $\alpha_f$ ensures that $\mH$ is a control barrier function according to
	Definition~\ref{def:barrier_function}, the previous discussion
	suggests selecting even larger values for $\alpha_f$ up to numerical limitations
	for solving~\eqref{eq:opt_BF}.
	%% Updated visualization of sets
	\begin{figure}
		\centering
		\begin{tikzpicture}[scale=0.9]
			\input{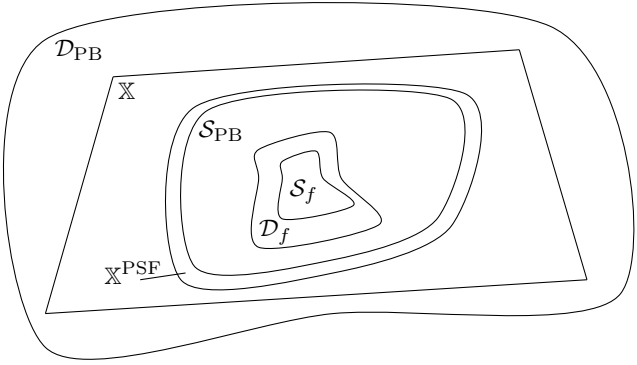}
		\end{tikzpicture}
		\caption{Overview of the different sets:
		The set $\XX$ represents the state constraints according to
		\eqref{eq:state_constraints} and $\XX^{PSF}$ represents the
		safe set of the hard constrained predictive safety filter problem~\eqref{eq:feasible_set_psf}
		with terminal safe set $\SS_f$ as described in the preliminaries,
		Section~\ref{subsec:psf}. Due to the softening of the state constraints in the
		predictive control barrier problem~\eqref{eq:opt_BF}, we obtain an
		enlarged region of attraction $\DD_{PB}$ according to
		\eqref{eq:domain_BF}, from which we can recover infeasible
		states. Due to the constraint tightening in~\eqref{eq:opt_BF_state_cons},
		the target set $\SS_{PB}$ for recovery is only a subset of the original hard
		constrained feasible set $\XX^{PSF}$. Finally, the set $\DD_f$ corresponds to the region of
		attraction of the terminal safe set $\SS_f$, used as terminal constraint
		in~\eqref{eq:opt_BF_terminal}.
		% , which is important in case of active
		% terminal slacks $\xi_N\neq 0$.
		}
		\label{fig:set_visualization_analysis}
	\end{figure}
	\begin{theorem}\label{thm:BF}
		Consider the predictive control barrier function $\mH$ as defined in~\eqref{eq:opt_BF}
		and assume that~$\UU$ and $\XX_{0}(\xi)$ as defined in~\eqref{eq:def_tight_state_constraints}
		are compact for all $0\leq \xi < \infty$. If Assumption~\ref{ass:terminal_cbf} holds
		with $\SS_f\subset\XX_{N-1}(0)$ and $h_f$ continuous on $\RR^n$, then
		the minimum~\eqref{eq:opt_BF} exists and for $\alpha_f<\infty$ large
		enough~ it follows that $\mH$ is a control barrier function
		according to Definition~\ref{def:barrier_function} with domain $\DD_{\mathrm{PB}}$ and
		safe set $\SS_{\mathrm{PB}}$.
	\end{theorem}
	\begin{proof}
		The existence of the minimum~\eqref{eq:opt_BF} is shown in Lemma~\ref{lem:existence_bf} in Appendix~\ref{app:technical_lemmas}.
		The remaining proof is structured in three different parts according to the
		properties required by Definition~\ref{def:barrier_function} as follows:
		In the first part, we first establish positive
		definiteness of $\mH(x)$ around $\SS_{\mathrm{PB}}$ in $\DD_{\mathrm{PB}}$.
		The second part establishes local
		continuity of $\mH(x)$ in $\DD_{\mathrm{PB}}$ and~compactness of
		$\SS_{\mathrm{PB}}$ and $\DD_{\mathrm{PB}}$.~
		In the last part we use a feasible solution at the current time step
		to construct a sub-optimal candidate solution to~\eqref{eq:opt_BF}
		at the next time step, for which we derive a sufficiently large bound on $\alpha_f$ that
		implies~ forward invariance of $\SS_{\mathrm{PB}}$ and $\DD_{\mathrm{PB}}$~and the existence
		of a continuous decrease $\Delta\mH(x)$ between two time-steps.
		In the following, we denote the optimal sequence
		$\mUpredOpt_{i|k}$, $\mXpredOpt_{i|k}$, and $\xi^*_{i|k}$
		for $x_{0|k}=x$ as $\mUpredOpt_{i}(x)$, $\mXpredOpt_{i}(x)$, and
		$\xi^*_{i}(x)$. We sometimes only refer to $\mUpredOpt_i(x)$ as
		$\mXpredOpt_{i}(x)$ and $\xi^*_{i}(x)$ can be defined correspondingly.
		\par		
		\paragraph{Positive definiteness of $\mH$ around $\SS_{\mathrm{PB}}$ in $\DD_{\mathrm{PB}}$}
		\label{proof:BF_lpd}
		By definition of $\SS_{\mathrm{PB}}$ it follows	directly that $\mH(x) = 0$ for all
		$x\in\SS_{\mathrm{PB}}$. If $x\notin\SS_{\mathrm{PB}}$ then	there must exist a
		$\xi^*_{i}(x)$ for some $0\leq i \leq N$ such that $\Vert \xi^*_{i}(x)\Vert >0$
		and therefore by definition of the cost~\eqref{eq:opt_BF_cost} it follows that
		$\mH(x)>0$ for all $x\in \DD_{\mathrm{PB}} \setminus \SS_{\mathrm{PB}}$.
		\par
		\paragraph{Continuity of $\mH$ for all $x\in\DD_{\mathrm{PB}}$~\eqref{eq:domain_BF}
		and compactness of $\DD_{\mathrm{PB}}$ and $\SS_{\mathrm{PB}}$}
		\label{proof:BF_cont}
		We show that for every $\epsilon >0$ there exists a $\delta >0$ such that for
		any $x,\bar x\in\DD_{\mathrm{PB}}$ the condition $\mNormGen{x-\bar x}<\delta$ implies that
		$|\mH(x)-\mH(\bar x)|<\epsilon$. As an intermediate step towards showing continuity,
		we derive a suboptimal solution to~\eqref{eq:opt_BF} at state $\bar x$ based
		on an optimal input sequence $\mUpredOpt_{i}(x)$ at state $x$.
		Therefore, define the constant input sequence $\mUpredCand_{i}(\bar x) \mDef \mUpredOpt_{i}(x)$ for a given $x$
		and define the corresponding state predictions $\mXpredCand_{i}(\bar x)$ based
		on an initial state $\bar x_{0}=\bar x$ according to the
		dynamics~\eqref{eq:system} with corresponding slacks $\bar \xi_{i}(\bar x) \mDef \max(0,c(\mXpredCand_{i}(\bar x))+\Delta_i\mOnes{})$ and $\bar\xi_{N}(\bar x)=\max(0,h_f(\mXpredCand_{N}(\bar x)))$,
		satisfying~\eqref{eq:opt_BF_init}-\eqref{eq:opt_BF_terminal} for $\bar x_{0}=\bar x$
		by construction.
		% \par
		Notice that this candidate solution will be contained in a compact set
		for all $\bar x\in\DD_{\mathrm{PB}}$ and $\{\bar u_{i}\}_{i=0}^{N-1}\in\UU^{N-1}$ due to
		Lemma~\ref{lem:compact_bounding_set}, which implies uniform continuity
		of the dynamics $f$ and all constraint functions $c_j$, and $h_f$ in the following analysis
		according to the Heine-Cantor theorem \cite[Theorem 4.19]{rudin1964principles}.
		\par
		Since compositions of uniformly continuous functions yield a uniformly continuous function
		we have that any predicted state
		$\mXpredCand_{i}(\bar x)=f(f(..f(f(\bar x,\mUpredCand_{0}(x)),\mUpredCand_{1}(x))...),\mUpredCand_{N-1}(x))$
		is uniformly continuous in the initial condition $\bar x$. Next, we note that the objective~\eqref{eq:opt_BF_cost}
		corresponding to the constructed sub-optimal solution denoted by $\bar {h}_{\mathrm{PB}}(\bar x)$ is uniformly continuous in $\bar x$ due to
		uniform continuity of $c_j$, $f$, $h_f$, and the fact that the maximum and sum of uniformly continuous functions
		are uniformly continuous. Since the optimal solution is smaller or equal than the candidate solution, it holds that
		for every $\epsilon >0$ there exists a uniform $\delta>0$ such that $||x-\bar x||<\delta$ implies
		$\mH(\bar x)-\mH(x)\leq\bar {h}_{\mathrm{PB}}(\bar x)-\mH(x)<\epsilon$. Due to uniform continuity, it also follows for
		$||\tilde x - \tilde {\bar x}||<\delta$ with optimal solution $u^*_i(\tilde x)$ at
		$\tilde x\mDef \bar x$ and corresponding suboptimal solution at $\tilde{\bar x}\mDef x$ that
		$\mH(\tilde{\bar x})-\mH(\tilde x)<\epsilon$, which shows continuity of~\eqref{eq:opt_BF} in $x\in\DD_{\mathrm{PB}}$.
		\par
		To show compactness of $\DD_{\mathrm{PB}}$ and $\SS_{\mathrm{PB}}$, we notice that
		Lemma~\ref{lem:bounded_solution} implies boundedness of these sets. 
		Since $\mH(x)$ is non-negative and continuous on $\DD_{\mathrm{PB}}$ and since
		$\SS_{\mathrm{PB}}\subseteq \DD_{\mathrm{PB}}$,
		we can further conclude that the sets $\DD_{\mathrm{PB}}$ and $\SS_{\mathrm{PB}}$ are closed, since they are defined as pre-images of the closed sets $[0,\alpha_f\gamma_f]$ and $\{0\}$ with respect to the mapping $\mH(x)$, implying compactness.
		\par
		\paragraph{Forward invariance of $\DD_{\mathrm{PB}}$ and $\SS_{\mathrm{PB}}$ and decrease $\Delta h_{\mathrm{PB}}$ around $\SS_{\mathrm{PB}}$ in $\DD_{\mathrm{PB}}$}
		\label{proof:BF_decrease}
		We show~\eqref{eq:def_barrier_function_12} by constructing
		a potentially sub-optimal $\hat u(x)$ for every $x\in\DD_{\mathrm{PB}}$
		that implies the existence of a locally positive definite function $\Delta \mH(x)$ around
		$\SS_f$ in $\DD_{\mathrm{PB}}$. Note that for every
		$x\in\DD_{\mathrm{PB}}$ there exists an optimal solution $\mUpredOpt_i(x)$
		$i=0,..,N-1$ to~\eqref{eq:opt_BF} with corresponding state sequence
		$\mXpredOpt_i(x)$ and slack sequence $\xi_i^*(x)$. In the remainder
		of the proof we will establish that $\hat u(x)=\mUpredOpt_0(x)$
		yields the desired implications.
		As a first step, consider~\eqref{eq:def_barrier_function_2}. For all
		$x\in\SS_{\mathrm{PB}}$ it follows from
		\eqref{eq:opt_BF_cost} that $\mXpredOpt_{N|k}\in\SS_f$
		and we construct
		\begin{align*}
			&u^+_i(x)\in
			\begin{cases}
				\{\mUpredOpt_{i+1}(x)\},~i=0,..,N-2,\\
				K_{\mathrm{CBF},f}(\mXpredOpt_N(x)),~i=N-1,
			\end{cases} \\
			&x^+_i(x)=
			\begin{cases}
				\mXpredOpt_{i+1}(x),~i=0,..,N-1, \\
				f(\mXpredOpt_N(x), u^+_{N-1}(x)), i = N,\text{ and}
			\end{cases}\\
			&\xi^+_i(x)=0\cdot\mOnes{}~~\forall i=0,..,N.
		\end{align*}
		In the following we show feasibility of $u_i^+,x_i^+$, and $\xi_i^+$
		with respect to \eqref{eq:opt_BF} at state $f(x, \hat u(x))=\mXpredOpt_1(x)=x_0^+(x)$:
		\begin{enumerate}
			\item $u_{N-1}^+(x) \in \UU$ since
				$u_{N-1}^+(x)\in K_{\mathrm{CBF},f}(\mXpredOpt_{N}(x))$ with
				$h_f(\mXpredOpt_{N}(x))\leq 0$
				and $u_{i}^+(x)\in\UU$ for all $i=0,..,N-2$ since $\mUpredOpt_{i+1}(x)$
				is part of a feasible solution.
			\item $\max(c(x^+_i(x)))=\max(c(\mXpredOpt_{i+1}(x)))\leq -\Delta_{i+1}\mOnes{}$
				for all $i=0,..,N-2$ by construction of $x^+_i$ and since $\Delta_{i+1}>\Delta_i$
				it follows $\max(c(x^+_i(x)))\leq -\Delta_{i}\mOnes{}$
				and therefore $x^+_i(x)\in\XX_i(0)$ for all $i=0,..,N-2$.
			\item $h_f(x_{N-1}^+)\leq 0$ implies $x_{N-1}^+\in\SS_f\subseteq\XX_{N-1}(0)$
				by Assumption~\ref{ass:terminal_cbf} through Theorem~\ref{thm:cbf}, \ref{item:thm_cbf_1}).
			\item $h_f(x_{N-1}^+)\leq 0$ and
					$u^+_{N-1}\in K_{\mathrm{CBF},f}(x^+_{N-1})$
				implies $x_{N}^+\in\SS_f$ and therefore
				$\xi_{N}^+=\max(0,h_f(x_{N}^+))= 0$.
		\end{enumerate}
		We have constructed a feasible solution with optimal
		value 0 for the 1-step prediction, proving~\eqref{eq:def_barrier_function_2} and forward invariance of
		$\SS_{\mathrm{PB}}$ under $u(k)=\mUpredOpt_{0|k}$.
		In the second step we show~\eqref{eq:def_barrier_function_1}, i.e.
		we consider the case $x\in\DD_{\mathrm{PB}}\setminus\SS_{\mathrm{PB}}$, implying by~\eqref{eq:opt_BF_cost} and the definition of $\gamma$ that $\mXpredOpt_{N|k}\in\DD_f$.~
		Let $x_i^+(x)$ and $u_i^+(x)$ be defined as above, where we omit the
		dependency on $x$ in the following. In addition, let
		\begin{align*}
			\xi^+_i =
			\begin{cases}
				\max(0,\xi^*_{i+1} + (\Delta_i - \Delta_{i+1})\mOnes{}),~ i=0,..,N-2,\\
				\max(0, c(x^+_{N-1})),~~i=N-1,\\
				\max(0, h_f(x^+_{N})),~~i=N.
			\end{cases}
		\end{align*}
		Note that feasibility of $\xi_{N-1}^+$ and $\xi_{N}^+$ w.r.t.~\eqref{eq:opt_BF}
		is given by definition. For $i=0,..,N-2$ we have
		\begin{align*}
			\max(0,c(x^+_i)) + \Delta_{i+1}\mOnes{} &\leq \xi_{i+1}^*\\
			\Leftrightarrow \max(0,c(x^+_i)) + \Delta_{i}\mOnes{} &\leq  \xi_{i+1}^*- \Delta_{i+1}\mOnes{} + \Delta_{i}\mOnes{} \\
			\Rightarrow \max(0,c(x^+_i)) + \Delta_{i}\mOnes{}&
			\leq
			\underbrace{\max(0,\xi_{i+1}^*- \Delta_{i+1}\mOnes{} + \Delta_{i}\mOnes{})}_{=\xi_i^+}
		\end{align*}
		ensuring feasibility of $\xi^+_{i}$ as well. After establishing
		feasibility of the candidate sequence, we use $\xi^+_i$ to bound the
		cost decrease
		\begin{align*}
			\Delta\mH(x)\mDef &\mH(x_0^+)-\mH(x)\\
			&\leq\underbrace{\alpha_f\left(\xi^+_N-\xi^*_N\right)+\Vert\xi^+_{N-1}\Vert}_{=:H_1} \\
			&\quad\underbrace{-\Vert \xi^*_0\Vert +\sum_{i=0}^{N-2}\Vert\xi^+_i\Vert-\Vert\xi^*_{i+1}\Vert}_{=:H_2}
		\end{align*}
		by first noting that $H_2\leq 0$ follows directly by construction since
		$\Delta_i < \Delta_{i+1}$. For $H_1$ we distinguish three possible cases:
		\begin{enumerate}
			\item $x^+_{N-1}=x^*_N\notin \XX_{N-1}(0)$, implying $x^+_{N-1}\notin \SS_f$.
				Let $\hat c = \max_{x\in\DD_f}\Vert \max(0,c(x)) \Vert$ be the maximum
				attainable norm of positive values of the constraint functions~\eqref{eq:state_constraints}
				in the terminal domain that is finite due to continuity of the constraints $c$ and boundedness
				of $\DD_{f}$. In this case we have
				$\xi^+_N=\max(0,h_f(x^+_N))$ and $\xi^*_N=h_f(\mXpredOpt_{N})$ and therefore
				\begin{align*}
					 &\alpha_f\left(\xi^+_N-\xi^*_N\right)+\Vert\xi^+_{N-1}\Vert
					 \leq-\alpha_f\Delta \xi_f(\mXpredOpt_{N}) + \hat c
				\end{align*}
				with the continuous (compare with \eqref{eq:lyapunov_decrease}) decrease defined as
				\begin{align}
					\Delta \xi_f(\mXpredOpt_{N})
					=\max(0,\min(\Delta h_f(\mXpredOpt_{N}),\max(0,h_f(\mXpredOpt_{N})))).
					\label{eq:pcbf_decrease}
				\end{align}
				with the continuous terminal decrease bound $\Delta h_f(x)$ from
				Assumption~\ref{ass:terminal_cbf}.
				Since $\DD_f$ and $\XX_{N-1}(0)$ are compact
				sets it follows for $x^*_N\notin \XX_{N-1}(0)$ and
				$x^*_N\in \DD_f$  that the smallest possible
				value $\Delta \xi_f(x^*_N)$ is lower bounded by $\min_{x\in \DD_f \setminus \mathrm{int}(\XX_{N-1}(0))} \Delta \xi_f(x)=\epsilon_f$. Due to $\SS_f\subset\XX_{N-1}(0)$,
				$x^*_N\notin\XX_{N-1}(0)\Rightarrow x^*_N\notin\SS_f$,
				and since \eqref{eq:pcbf_decrease} is larger than zero for
				$x\notin \SS_f$, it holds that $\epsilon_f>0$. This implies for all $\mXpredOpt_{N}\notin \XX_{N-1}(0)$ that
				$\Delta \xi_f(\mXpredOpt_{N})>\epsilon_f$ holds. Selecting
				$\alpha_f\geq\hat c \epsilon_f^{-1}$ therefore implies
				that the term $H_1$ is strictly less than zero  in this case.
			\item $x^+_{N-1}= \mXpredOpt_{N} \in \XX_{N-1}(0)$ and $x^+_{N-1}\notin \SS_f$. In this
				case it follows directly that $\xi_{N-1}^+=0$ and
				$\alpha_f\left(\xi^+_N-\xi^*_N\right)\leq -\alpha_f\Delta \xi_f(\mXpredOpt_{N})$
				and therefore $H_1<0$.
			\item $x^+_{N-1} = \mXpredOpt_{N}\in \XX_{N-1}(0)$ and
				$x^+_{N-1}\in \SS_f$. In this case we have
				$H_1=0$. However, note that $x\in\DD_{\mathrm{PB}}\setminus\SS_{\mathrm{PB}}$
				implies that there must exist a $j$, $0\leq j \leq N-1$ such that
				$\Vert\xi_j^*\Vert > 0$ and by definition of $\xi^+_i$ that
				$\Vert\xi_{j-1}^+\Vert - \Vert\xi_{j}^*\Vert < 0$ for $1\leq j\leq N-1$
				since $\Delta_{i+1} > \Delta_i$ and therefore $H_2<0$.
		\end{enumerate}
		We have therefore shown that for any possible state or terminal
		constraint softening given $x^+_{N-1}=x_{N}^*$, either $H_1<0$
		or $H_2<0$ holds if $x\in\DD_{\mathrm{PB}}\setminus \SS_{\mathrm{PB}}$.
		Furthermore, similar as done in the first part of the proof, it
		can be shown that all introduced bounds on $\Delta h_{\mathrm{PB}}$ are
		continuous due to continuity of the dynamics $f$ and $\Delta h_f$,
		implying that there exists a continuous maximum $\Delta \mH$
		over these cases as required by Definition~\ref{def:barrier_function}~\eqref{eq:def_barrier_function_1}.
		From the definition of $\DD_{\mathrm{PB}}$ in~\eqref{eq:domain_BF}, the upper bound $\Delta h_{PB}$ implies forward invariance of $\DD_{\mathrm{PB}}$ under $u(k)=\mUpredOpt_{0|k}$.
		Together with continuity of $\mH$ from the first part of the proof,
		we conclude that $\mH$ is a control barrier function according to
		Definition~\ref{def:barrier_function}.
	\end{proof}
	The proof of Theorem~\ref{thm:BF} not only establishes that $\mH(x)$
	is a control barrier function according to Definition~\ref{def:barrier_function}
	but also shows that any input computed according to Algorithm~1, line 4 will be
	safe, i.e. that it satisfies $\mUpredOpt_{0|k}(x(k))\in K_{\mathrm{CBF}}(x(k))$,
	as defined in \eqref{eq:def_barrier_function_3}. This is ensured through the
	existence of a sub-optimal slack sequence, implying the required decrease $\Delta\mH(x(k))$
	between consecutive time steps for any input sequence $\mUpredOpt_{i|k}(x(k))$
	satisfying the constraints $\XX_i(\xi_i^*(x(k)))$ with $\xi^*(x(k))$ resulting
	from~\eqref{eq:opt_BF}.
	In the case of infeasible initial conditions, e.g. during system startup or due to
	large disturbances, Algorithm~1  ensures that
	potential constraint relaxations $\xi_i >0$ are stabilized and asymptotically converge to zero if no additional disturbances occur.
	As a result, a tightened version of the feasible set of the
	safety filter $\SS_{\mathrm{PB}}$~\eqref{eq:safe_set_BF} is stabilized and
	Algorithm~1 thereby recovers the system from constraint violations
	if $x(0)\in\DD_{\mathrm{PB}}$. We illustrate this effect for numerical examples in
	Section~\ref{sec:numerical_examples}.
	\par
	The proposed recovery mechanism for a PSF offers a number of 
	additional benefits. The predictive control barrier function can be
	used as a safety metric when combining safety filters with potentially unsafe learning-based
	controllers $u_p(k)$ through $\mH(f(x(k),u_p(k)))$, see, e.g., \cite{Fisac2019}, to accelerate the overall learning performance.
	Note that existing predictive safety filter schemes such as \cite{Wabersich2018} 
	only allow to penalize safety ensuring interventions $\Vert u_s(k) - u_p(k)\Vert$ at
	each time step, which is typically non-continuous w.r.t. the state $x(k)$ and the
	learning input $u_p(k)$	and typically does not relate to the actual `danger'
	that $u_p(k)$ might induce in terms of state constraint violations in the future.
	\par
	In addition, continuity of $\mH(x)$ enables the computation of explicit
	approximations $\hat{h}_{\mathrm{PB}}(x)$ of $\mH(x)$ within $\DD_{\mathrm{PB}}$ using, e.g.,
	deep learning techniques, which can significantly reduce the online computational burden by replacing
	Algorithm~1 line 3,4 with $u_{0|k} \in \argmin_{u\in\RR^n}
	\mNormGenSmall{u_p(x(k))-u} \text{ s.t. } \hat{h}_{\mathrm{PB}}(f(x(k),u)) - \hat{h}_{\mathrm{PB}}(x(k))\leq -\Delta \hat h(x(k))$, where
	$\Delta \hat h(x(k)) = 0$ for $\hat{h}_{\mathrm{PB}}(x(k))=0$ and
	$\Delta \hat h(x(k)) < 0$ otherwise.
	\begin{remark}\label{rem:mpc}
		Note that the proposed concept in Algorithm~1
		together with the advantages from Theorem~\ref{thm:BF}
		can also be used to enhance nominal model predictive controllers
		by replacing the objective in the predictive safety filter
		problem $\mNormGenSmall{u_p(x(k)) - \mUpred_{0|k}}$ with an appropriate
		sum of stage-cost functions, e.g. $\sum_{i=0}^{N-1}\ell(x_{i|k},u_{i|k})$.
	\end{remark}
\section{Terminal control barrier function design}\label{sec:terminal_cbf_design}
In this section, we provide a principled design procedure for a terminal
control barrier function according to Assumption~\ref{ass:terminal_cbf}.
To this end, we assume that the dynamics~\eqref{eq:system} are twice continuously differentiable
and that there exists a steady-state at the origin $0 = f(0,0)$ with $0\in \mathrm{int}(\XX)$
and $0\in\mathrm{int}(\UU)$, which allows to locally approximate~\eqref{eq:system} using
$x(k+1)=A x(k+1)+B u(k) + r(x(k),u(k))$ with linearized dynamics
$A\mDef (\partial/\partial x) f(x,u)|_{(0,0)}$, $B\mDef (\partial/\partial u) f(x,u)|_{(0,0)}$,
and higher-order error terms $r(x(k), u(k))= f(x,u) - Ax-Bu$. We additionally assume that the state and
input constraints are convex polytopes of the form $\XX=\{x\in\RR^n | A_x x \leq b_x\}$,
$A_x\in\RR^{n_x\times n}$, $b_x\in\RR^{n_x}$ and $\UU=\{u\in\RR^m | A_u u \leq b_u\}$,
$A_u\in\RR^{n_u\times n}$, $b_u\in\RR^{n_u}$ and that a corresponding constraint tightening according
to the definition~\eqref{eq:def_tight_state_constraints} has been selected, e.g. using $\Delta_i = i\cdot(1/N)c_\Delta$
for some design parameter $c_\Delta \in (0,1)$.
\par
Following similar design steps as presented, e.g., in~\cite{chen1998quasiInfiniteHorizonMPC,rawlings2009model},
we restrict our attention to a quadratic terminal control barrier function of the form
$h_f(x)=x^\top P x -\gamma_x$, $P\in\mSetPosSymMat{n}$, $\gamma_x>0$, with safe set
$\SS_f = \{x\in\RR^n | h_f(x)\leq 0\}$, domain $\DD_f=\{x\in\RR^n | h_f(x)\leq \gamma_f\}$,
and quadratic decrease $-\Delta h_f(x)=-\mu_x x^\top x - \mu_u u^\top u$ to satisfy
Definition~\ref{def:barrier_function} in the case $r(x,u)=0$, which we adjust
to the nonlinear case $r(x,u)\neq 0$ afterwards.
To this end, we explicitly parametrize a linear control law of the form $u = Kx$, with
$K\in\RR^{m\times n}$, enabling application of convex optimization techniques for designing
$P$ with the goal of obtaining a possibly large domain $\DD_f$ and safe set $\SS_f$. Thereby, specific
values for $\mu_x>0$ and $\mu_u>0$ can be selected to trade off a compensation of the linearization
error $r(x,u)$ against the resulting aggressiveness of the state feedback controller.
This results in the following design:
\paragraph*{Step 1} Select $\mu_x,\mu_u>0$, define $P=E^{-1}$, $K=YE^{-1}$ (see, e.g., \cite{Boyd1994})
and solve
\begin{subequations}\label{eq:opt_design_sdp}
\begin{align}
	\min_{\substack{E\in\mSetPosSymMat{n},\\{Y\in\RR^{m\times n}}}} \quad & -\mathrm{logdet}(E) \label{eq:opt_design_sdp_obj}\\
	\text{s.t.}\quad & E \succeq 0, \\
	&\begin{bmatrix}
		E & EA^\top+Y^\top B^\top & E I_n \mu_x & Y^\top I_m \mu_u \\
		* & E & 0 & 0 \\
		* & * & I_n & 0 \\
		* & * & * & I_m
	\end{bmatrix} \succeq 0 \label{eq:opt_design_sdp_inv}\\
	& \begin{bmatrix}
		b_{u,j}^2 & A_{u,j}Y \\
		* & E
	\end{bmatrix}\succeq 0, ~~\forall j=1,..,n_u
	\label{eq:opt_design_sdp_cons}
\end{align}
\end{subequations}
with $*$ defining transposed elements. The objective~\eqref{eq:opt_design_sdp_obj}
maximizes the volume of the domain $\DD_f$ and the safe set $\SS_f$,
\eqref{eq:opt_design_sdp_inv} enforces the desired decrease $-\Delta h_f$, and
\eqref{eq:opt_design_sdp_cons} ensures input constraint satisfaction
under $u=Kx$ for all $x\in\DD_f$ with $\gamma_f=1-\gamma_x$.
The terminal safe set $\SS_f$ is obtained through the support values of the state
constraint half-spaces as $\gamma_x = \min(1-c,\min_j~(b_{x,j}-\Delta_{N-1})^2(A_{x,j} P^{-1}A_{x,j}^\top)^{-1})$
for some small $c>0$ to ensure $\SS_f\subset\DD_f\subseteq\XX$ and constraint tightening
$\Delta_{N-1}$. The resulting control
barrier function $h_f(x)=x^\top P x - \gamma_x$ therefore satisfies the required
properties according to Definition~\ref{def:barrier_function} for the linearized system, i.e.
for the case $r(x,u)=0$ by construction.
\paragraph*{Step 2} As shown in \cite[Section 2.5.5]{rawlings2009model}, feasibility
of~\eqref{eq:opt_design_sdp} guarantees that there exist valid, non-empty sub-domains of
$\DD_f$ and $\SS_f$ for the nonlinear system~\eqref{eq:system}. These sub-domains can be found by 
solving a verification problem for a specific choice of $\gamma_f$ and $\gamma_x$,
which can be iteratively reduced, if necessary.
Invariance of $\SS_f$ can be verified via
\begin{subequations}\label{eq:design_nonlinear_verification_step3}
\begin{align}
	0\geq \max_{x\in\RR^n} \quad & h_f(f(x,Kx)) \\
	\text{s.t.} \quad & h_f(x)\leq 0
\end{align}
\end{subequations}
using nonlinear programming. If this condition does not hold,
decrease $\gamma_x$ and repeat \emph{Step 2}.
\paragraph*{Step 3} 
To determine $\DD_f$, initialize $\gamma_f = 1-\gamma_x$ and verify
\begin{subequations}\label{eq:design_nonlinear_verification}
\begin{align}
	0>\max_{x\in\RR^n} \quad & h_f(f(x,Kx)) - h_f(x) \label{eq:design_nonlinear_verification1}\\
	\text{s.t.} \quad & 0 < h_f(x)\leq \gamma_f. \label{eq:design_nonlinear_verification2}
\end{align}
\end{subequations}
If~\eqref{eq:design_nonlinear_verification1} does not hold then
decrease $\gamma_f\in[\gamma_x,1]$ and repeat \emph{Step 2}.

Note that \emph{Step 1} and \emph{Step 2} require a global solution
of the potentially non-convex optimization problems \eqref{eq:design_nonlinear_verification_step3}
and \eqref{eq:design_nonlinear_verification}. In non-convex cases, a practical strategy is
to randomly select different warm-starts for the underlying optimization algorithm to
obtain a local optimum, which reflects the global optimum with high probability,
see, e.g., \cite{Boender1995}.
The resulting function $h_f$ allows for setting up the predictive
barrier function problem~\eqref{eq:opt_BF} with
terminal penalty scaling factor $\alpha_f=\hat c\epsilon_f^{-1}$, $\hat c = \max_{x\in\DD_f}\Vert \max(0,c(x)) \Vert$,
$\epsilon_f=\min_{x\in \DD_f \setminus \mathrm{int}(\XX_{N-1}(0))}\Delta \xi_f(x)$ according to
the proof of Theorem~\ref{thm:BF}, \eqref{eq:pcbf_decrease}.
If $\DD_f \setminus \mathrm{int}(\XX_{N-1}(0))\neq\emptyset$, then the first
term $H_1$ in the proof of Theorem~\ref{thm:BF} vanishes,
and we can select an arbitrary $\alpha_f>0$, possibly large as
discussed in the paragraph before Theorem~\ref{thm:BF}.

\section{Numerical examples}\label{sec:numerical_examples}
In this section, we first illustrate the proposed predictive barrier function approach using a
small-scale linear system, where we recover a stabilizing linear controller from initial
state constraint violations. In addition, we consider a nonlinear vehicle
model, controlled by an MPC, which encounters unusual large disturbances leading to infeasibility,
which will be recovered through the proposed method. For set computations we use
YALMIP~\cite{Lofberg2004} and the MPT toolbox~\cite{mpt3} and for implementation
of the predictive control barrier and model predictive control problem, we used the
Casadi~\cite{Andersson2018} framework for automatic differentiation together
with the nonlinear optimization solver IPOPT~\cite{biegler2009large}.

\subsection{Illustrative unstable linear example}\label{subsec:linear_example}
\begin{figure*}
	\includegraphics[width=0.45\linewidth]{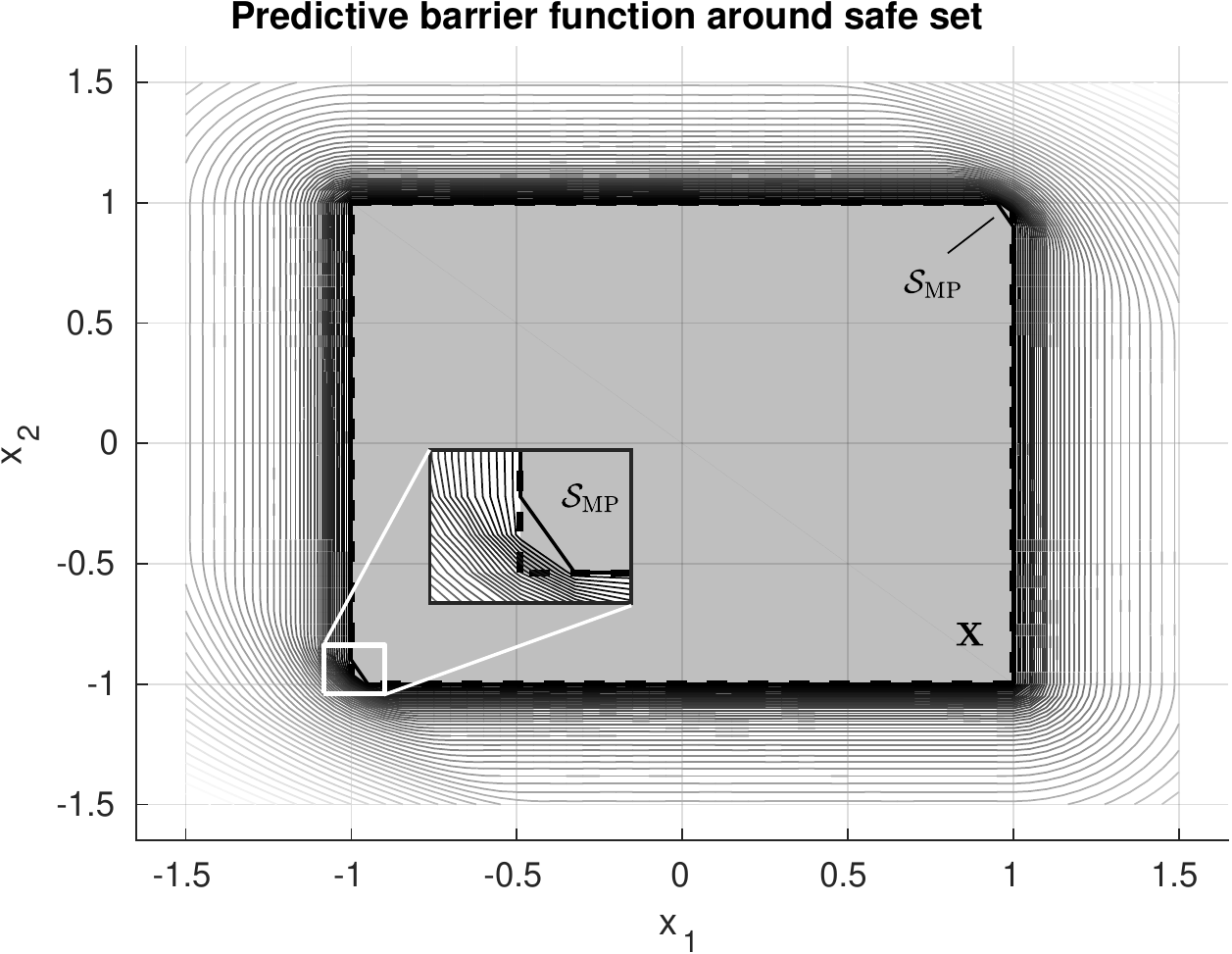}
	\hfill
	\includegraphics[width=0.45\linewidth]{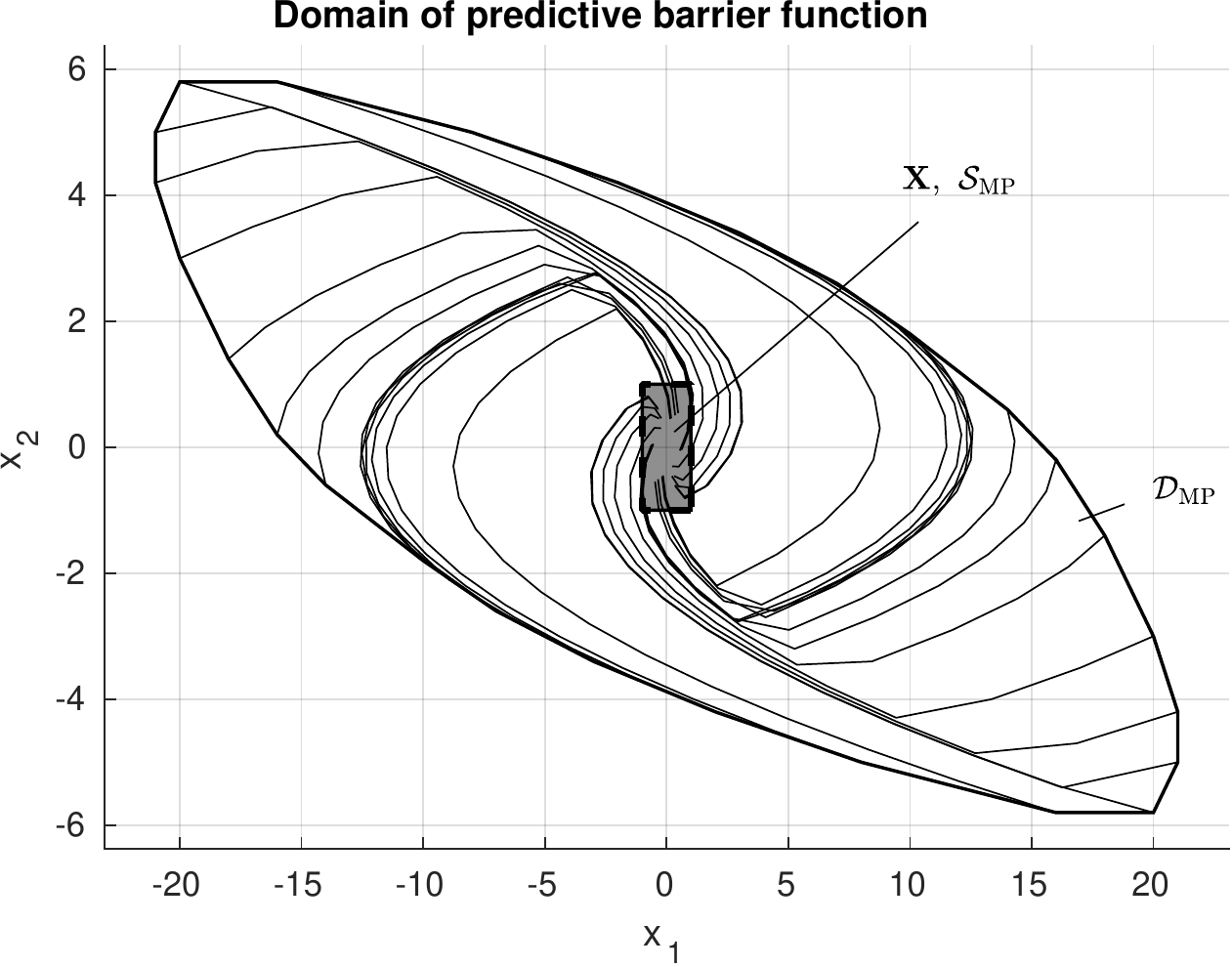} \\
	\includegraphics[width=0.45\linewidth]{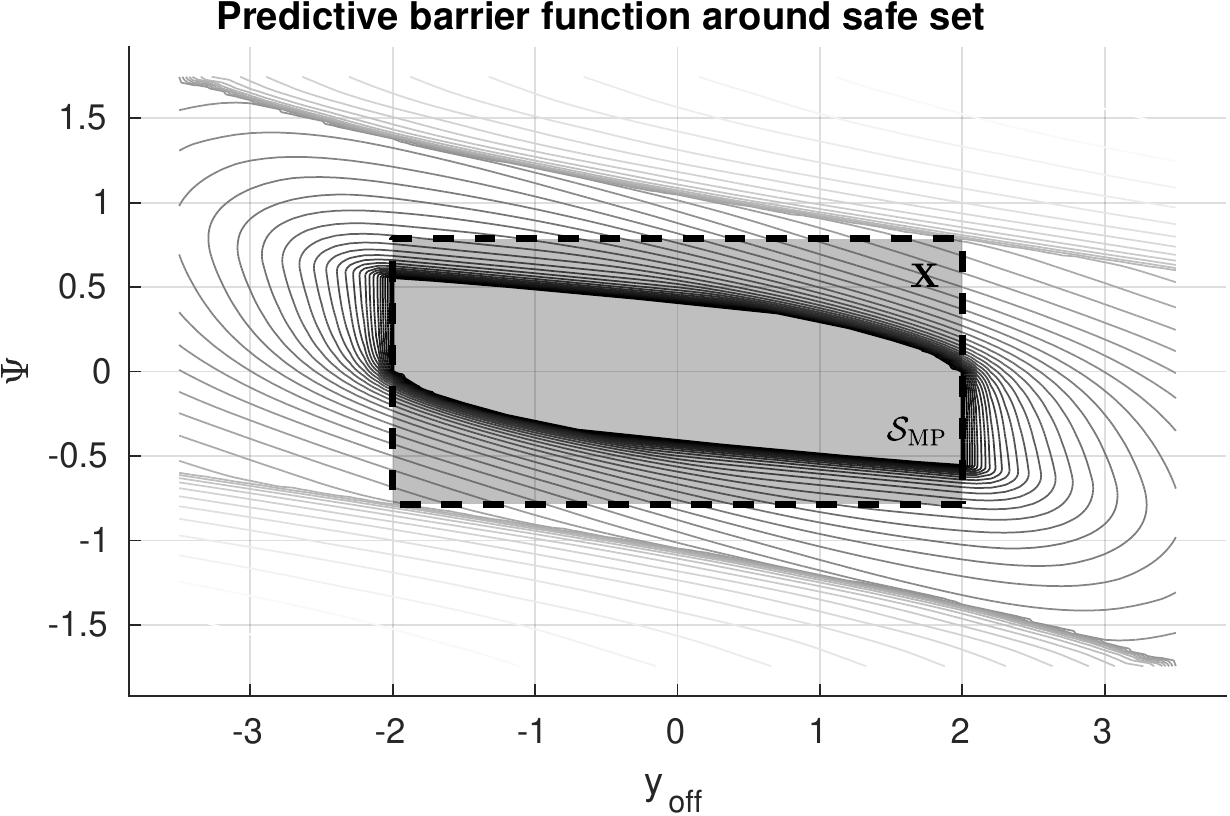}
	\hfill
	\includegraphics[width=0.45\linewidth]{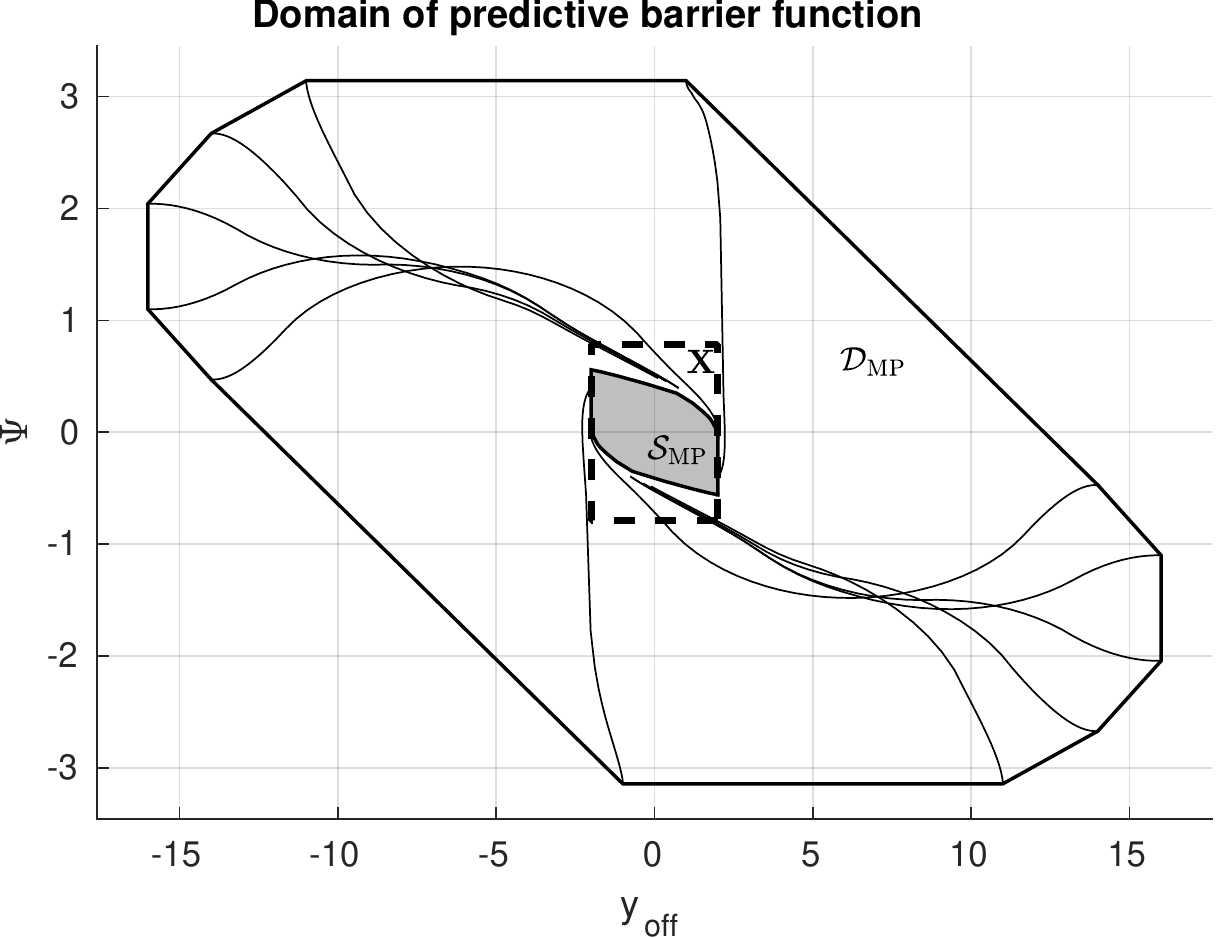}
	\caption{Illustrative example~\ref{subsec:linear_example} (\textbf{top row})
	and the nonlinear vehicle example according to Section~\ref{subsec:nonlinear_example}
	by projecting onto the states $y_{\mathrm{off}}$ and $\Phi$ (\textbf{bottom row}).
	\textbf{Left column:} Shown is the state space $\XX$ together with predictive safe
	set $\SS_{\mathrm{PB}}$ as defined in Theorem~\ref{thm:BF}. Gray contour
	lines depict level sets of the predictive barrier function $\mH(x)$ around
	the safe set $\SS_{\mathrm{PB}}$. \textbf{Right column:} The outer set $\DD_{\mathrm{PB}}$
	represents the domain of $\mH(x)$ as defined in~\eqref{eq:domain_BF} with simulated
	closed-loop recovery trajectories towards the safe set $\SS_{\mathrm{PB}}\subseteq\XX$
	under application of $u(k)=\mUpredOpt_{0|k}(x(k))$.}
	\label{fig:numerical_examples_domain}
\end{figure*}
Consider the unstable linear system~\cite{Zeilinger2014}
\begin{align}\label{eq:unstable_linear_system}
	x(k+1) =
	\begin{bmatrix}
		1.05 & 1 \\
		0 & 1
	\end{bmatrix}
	x(k) + 
	\begin{bmatrix}
		1 \\ 0.5
	\end{bmatrix}
	u(k),
\end{align}
which is subject to state and control constraints
$\Vert x \Vert_\infty\leq 1$ and $\Vert u \Vert_\infty \leq 1$.
We are given a simple stabilizing performance controller of the form $u(k)=K_p x(k)$
with $K_p=[-0.5592,-0.9445]$ and address recovery of the system from initial state
constraint violations $x\notin\XX$, which can additionally cause infeasible control
inputs $u=K_px\notin\UU$.
\par
The constraint tightening is selected as $\Delta_i = i\cdot 0.005$~along prediction
time steps $i=0,..,19$. To design the corresponding terminal control
barrier function according to Assumption~\ref{ass:terminal_cbf} and
Theorem~\ref{thm:BF}, we apply the procedure described in
Section~\ref{subsec:terminal_configuration} with $\mu_x=\mu_u=0.01$
to obtain $P=\left[\begin{smallmatrix} 0.007 & 0.015 \\ 0.015 & 0.086\end{smallmatrix}\right]$,
$K=-[0.062,~0.288]$ with corresponding level sets $\gamma_f=1$ and $\gamma_x=0.0035$.
The resulting terminal control barrier function is therefore given
by $h_f = \max(0, x^\top P x -\gamma_{x})$
with domain $\DD_f=\{x\in\RR^n: x^\top P x \leq 1\}$. The predictive
control barrier function problem~\eqref{eq:opt_BF} is implemented
using a planning horizon $N = 20$ with terminal weight $\alpha_f = 1000$,
satisfying the bound in the proof of Theorem~\ref{thm:BF},
yielding the domain $\DD_{\mathrm{PB}} = \{x\in\RR^n:\mH(x) < 1000\}$ according
to~\eqref{eq:domain_BF}. 

In Figure~\ref{fig:numerical_examples_domain} (top left),
we show the state constraints together with the resulting
predictive safe set $\SS_{\mathrm{PB}}=\{ x\in\RR^n: \mH(x) = 0\}$ according to
Theorem~\ref{thm:BF}, i.e. the zero level set of the
predictive barrier function $\mH(x)$~\eqref{eq:opt_BF}
together with logarithmically scaled contour lines of $\mH(x)$
around $\SS_{\mathrm{PB}}$. Despite the constraint
tightening, the safe set $\SS_{\mathrm{PB}}$ covers almost
the entire state space, providing a possibly small
amount of interference with respect to the linear control
law inside the state constraints $\XX$.

We demonstrate recovery from infeasible initial
conditions in Figure~\ref{fig:numerical_examples_domain} (top right),
where we display the enlarged domain $\DD_{\mathrm{PB}}$ of $\mH(x)$ as defined
in~\eqref{eq:domain_BF} together with the safe set
of the predictive control barrier function
$\SS_{\mathrm{PB}}$ and closed-loop system trajectories under application
of $u(k)=\mUpredOpt_{0|k}(x(k))$, starting at the
boundary of $\DD_{\mathrm{PB}}$. As guaranteed by
Theorem~\ref{thm:BF}, all trajectories converge to
$\SS_{\mathrm{PB}}$.

While the presented predictive control barrier function
method has some fundamental differences compared with
the soft constrained approach presented in~\cite{Zeilinger2014} that is limited
to linear systems and provides stability with respect to the
origin only, the resulting enlarged region of attraction
$\DD_{\mathrm{PB}}$ has a similar shape.

\subsection{Nonlinear example: Kinematic car model}\label{subsec:nonlinear_example}
In this example, the goal is to recover an unsafe state of a kinematic
vehicle model from large initial deviations with respect to lateral
safety constraints. For the car simulation we consider the dynamics
\begin{align}
	\dot y_{\mathrm{off}} &= (v_s + v) \sin(\Psi) \\
	\dot \Psi &= ((v_s + v)/L)\tan(\delta) \\
	\dot \delta &= u_1 \\
	\dot v & = u_2,
\end{align}
where $y_{\mathrm{off}}$ and $\Psi$ define the offset and relative angle
with respect to a centerline, $\delta$ the steering angle,
$v$ the relative longitudinal vehicle speed with respect to a
target velocity $v_s= 5~\mathrm{[m/s]}$, $u_1$ the applied steering rate,
and $u_2$ the applied acceleration. The physical input limitations
are given by $|u_1|\leq 1.4$ and $-5\leq u_2 \leq 2$ as well as
$|\delta| \leq 0.35$. The desired safety constraints are defined
as $|y_{\mathrm{off}}|\leq 2$, $|\Psi|\leq \pi/4$, and $-4\leq v\leq 5$. Similar
to the linear example, we consider recovery of an infeasible
initial condition $x\notin \XX$.

For the design of the required terminal barrier function $h_f$
according to Assumption~\ref{ass:terminal_cbf} and Theorem~\ref{thm:BF},
we proceed as described in Section~\ref{sec:terminal_cbf_design}
by first discretizing the system using Euler forward and a sampling interval of $0.05$ [s].
We select the constraint tightening $\Delta_i = i\cdot 0.004$ and linearize
around the origin to compute for $\mu_x=\mu_u=0.1$ the matrix
\begin{align*}
	P =
	\begin{bmatrix}
		0.30&1.07&0.39&0.00\\
		1.07&5.72&2.54&0.00\\
		0.39&2.54&1.69&0.00\\
		0.00&0.00&0.00&0.20
	\end{bmatrix},
\end{align*}
with corresponding controller gain
$K=-\left[\begin{smallmatrix}
	0.28&2.45&1.77& 0.00 \\
	 0.00& 0.00& 0.00&0.90
	\end{smallmatrix}\right]$.
Application of \textit{Step 2} and \textit{Step 3} in
Section~\ref{sec:terminal_cbf_design} yield $\gamma_f=0.13$ and $\gamma_x = 0.0108$.
The corresponding the terminal
control barrier function is given by $h_f(x)=x^\top P x-\gamma_x$ with
terminal weight $\alpha_f = 10^5$, satisfying the bound in
the proof of Theorem~\ref{thm:BF}.

In Figure~\ref{fig:numerical_examples_domain} (bottom left),
we plot the state constraints together with the resulting
predictive safe set $\SS_{\mathrm{PB}}$ with planning horizon $N=50$ for the states $\Psi$ and
$y_{\mathrm{off}}$ by setting $v=\delta=0$ together with logarithmically scaled
contour lines of $\mH(x)$ around $\SS_{\mathrm{PB}}$. The safe set
$\SS_{\mathrm{PB}}$ touches the lateral constraints only for
car headings that do not point away from the center line as expected.

In Figure~\ref{fig:numerical_examples_domain} (bottom right) we display the resulting domain
of the predictive control barrier function $\mH(x)$ as defined in~\eqref{eq:domain_BF},
from which we can recover infeasible initial conditions and provide sample closed-loop simulations
from the extreme points of $\DD_{\mathrm{PB}}$, which all converge to $\SS_{\mathrm{PB}}$
as desired.

\section{Conclusion}
This paper has addressed the problem of infeasibility of predictive
safety filters resulting, e.g., from infeasible initial system conditions
or large disturbances. Since a simple softening of the state constraints
does not necessarily imply recovery from constraint violations, we proposed
a recovery mechanism with an auxiliary feasibility problem using an iterative
constraint tightening along the planning horizon together with a terminal safe set,
which is required to be a level set of a corresponding discrete-time control barrier
function. Asymptotic stability of the feasible set of the original predictive
safety filter problem under the proposed algorithm is shown using ideas from control
barrier function theory. A principled design procedure for the required components
was provided together with numerical examples to demonstrate recovery from constraint
violations.

\bibliographystyle{IEEEtran}
\bibliography{bibliography.bib}

\appendix

\subsection{Lyapunov stability with respect to sets}\label{app:lyap_proof}
	Consider the discrete-time autonomous system of the
	form~\eqref{eq:autonomous_system} with dynamics $\mDefFunction{g}{\RR^n}{\RR^n}$
	and initial condition $x(0)=x_0$ with $x_0\in\RR^n$.  Our goal in this section is
	to show Lyapunov stability results in terms of a safe set $\SS\subset\RR$, e.g.
	given by~\eqref{eq:feasible_set_psf}, which is positively invariant.
	Following standard stability arguments, the result will be established via a
	Lyapunov function.
	\begin{definition}\label{def:pos_def_func}
		Let $\SS,\DD \subset \RR^n$ be non-empty and compact sets with
		$\SS\subset\DD$ and consider a continuous function $\mDefFunction{V}{\DD}{\RR}$.
		We call $V(x)$ a locally positive definite (l.p.d.) function
		around $\SS$ in $\DD$ if it holds that
		\begin{subequations}
			\begin{align}
				\forall x \in \SS:~V(x) &= 0 \text{ and} \\
				\forall x \in \DD\setminus\SS:~V(x)&>0.
			\end{align}
		\end{subequations}
		\END
	\end{definition}
	\begin{definition}\label{def:lyap_func}
		Let $\SS,\DD \subset \RR^n$ be non-empty and compact sets with
		$\SS\subset\DD$ and consider a l.p.d. function $\mDefFunction{V}{\DD}{\RR}$
		around $\SS$ in $\DD$. If there exists a continuous l.p.d. function
		$\Delta V(x)$ around $\SS$ in $\DD$ such that for all $x\in\DD$
		the difference inequality
		\begin{align}\label{eq:def_lyap_func_decrease}
			 V(g(x)) - V(x)\leq -\Delta V(x)
		\end{align}
		holds with respect to system~\eqref{eq:autonomous_system},
		then $V$ is called a Lyapunov function for
		system~\eqref{eq:autonomous_system} with respect to $\SS$ in $\DD$.
		\END
	\end{definition}
	Using a similar analysis structure as in the case of Lypunov stability
	analysis with respect to equilibrium points, we can extend existing
	results to also hold with respect to invariant sets without relying on the
	existence of lower and upper bounding class $\mathcal K$ functions on the
	Lyapunov function. While these are commonly used in MPC literature, see, e.g.,
	in~\cite[Appendix B.2]{rawlings2009model}, establishing existence of the required class $\mathcal K$
	functions for~\eqref{eq:opt_BF} would be difficult due to the lack of a positive definite
	stage cost function with respect to the implicit target set of states $\SS_{\mathrm{PB}}$
	as considered in Theorem~\ref{thm:BF}.
	\begin{theorem}\label{thm:lyapunov_stability}
		Consider system~\eqref{eq:autonomous_system}. Let
		$\SS$ and $\DD$ be non-empty and compact sets with
		$\SS \subset \DD$ that are positively invariant
		for system~\eqref{eq:autonomous_system}.
		If there exists a Lyapunov function with respect to
		$\SS$ in $\DD$ for system~\eqref{eq:autonomous_system},
		then $\SS$ is an asymptotically stable set for
		system~\eqref{eq:autonomous_system} in $\DD$.
	\end{theorem}
	\begin{proof}
	The following proof showing \eqref{eq:def_asy_stable} is based on~\cite[Theorem 4.1]{Khalil1996}
	with adjustments to account for a discrete-time setting as well as stability w.r.t. a
	set rather than an equilibrium point.
	Property~\eqref{eq:def_asy_stable_1}: We first consider the case $\epsilon>0$ such that for all
	$x \in \RR^n$ with $|x|_\SS \geq \epsilon$ it follows
	$x\notin \DD$. Due to forward invariance of $\DD$, this case
	fulfills~\eqref{eq:def_asy_stable_1} by definition for any
	$x(0)\in\DD$.
	\par
	In the remaining case, i.e. $\epsilon>0$ such that there exists
	an $x \in \RR^n$ with $|x|_\SS \leq \epsilon$ for which it holds $x\in\DD$,
	we construct a $\delta > 0$ in the following such that
	$|x(0)|_\SS < \delta \Rightarrow |x(k)|_\SS < \epsilon$ holds for all $k\in\NN$.
	Define
	\begin{align}\label{eq:proof_lyapunov_stability}
		V^-\mDef \min_{|x|_\SS\geq \epsilon,~x\in\DD}V(x),
	\end{align}
	the existence of which can be verified as follows: From $|x|_\SS$ being continuous
	we have that the pre-image $\{x~|~|x|_\SS\geq \epsilon\}$ is closed, which
	yields an overall compact subset constraint on $x$
	in~\eqref{eq:proof_lyapunov_stability} when intersecting the closed pre-image
	with the compact set $\DD$. Since $V(x)$ is continuous it follows
	from the extreme value theorem that the minimum $V^-$ in~\eqref{eq:proof_lyapunov_stability}
	exists and since $V(x)$ is l.p.d. around $\SS$ in $\DD$ with $\SS\subset \DD$ (Definition~\ref{def:pos_def_func})
	it follows that $V^- >0$.
	\par
	Due to continuity of $V$, we can select a $\delta > 0$ such
	that $\mNormGen{x-\bar x}<\delta \Rightarrow|V(x)-V(\bar x)|<V^-$ implying
	for all $\bar x\in\DD$ such that $|\bar x|_\SS < \delta$ that $V(\bar x) < V^-$.
	\par
	For an initial condition $|x(0)|_\SS<\delta$ and $x(0)\in\DD$
	we therefore have from $V(x(0))<V^-$ together with the fact that
	$V(x(k))$ is non-increasing for $k\in\NN$ (Definition~\ref{def:lyap_func})
	and positive invariance of $\DD$ and $\SS$ that $V(x(k))<V^-$ holds
	for all $k\in\NN$.
	\par
	Assume there exists a time step $\bar k\in\NN$ such that $|x(\bar k)|_\SS \geq \epsilon$.
	By construction of $V^-$ it follows $V(x(\bar k)) \geq V^-$, which is a contradiction
	to the statement $V(x(k))<V^-$ above and therefore proves property~\eqref{eq:def_asy_stable_1}.
	\par
	Property~\eqref{eq:def_asy_stable_2}: Compared to the first part of this proof,
	convergence can be shown along the lines of the proof of \cite[Theorem 4.1]{Khalil1996}.
	Let $x(0)\in\DD$.
	As established above, the sequence $V(x(k))$ is non-increasing
	for $k>0$ and $V$ is continuous on the bounded set $\DD$, implying
	that it will converge, i.e. $\lim_{k\rightarrow\infty}V(x(k))=\alpha$
	with $\alpha\geq 0$.
	\par
	For a proof by contradiction, select an $\alpha>0$ such that
	there exists a $x\in\DD$ implying $V(x)=\alpha$ and define
	$\AA\mDef\{x\in\DD:V(x)\leq \alpha\}$. Due to continuity
	of $V$, we can select a $\beta>0$ such that $|x|_\SS<\beta$
	implies $V(x)<\alpha$. It follows that the set
	$\BB\mDef\{x\in\DD:|x|_\SS<\beta\}$ with $\beta >0$
	satisfies $\BB\subset\AA$.
	\par
	Since $V(x(k))$ is monotonically decreasing to $\alpha$ by
	assumption (Definition~\ref{def:lyap_func}) it holds
	that $V(x(k))\geq \alpha$ for all $k>0$, and therefore we have
	$x(k) \notin \BB$ for all $k>0$.
	By excluding the possibility of $x(k)\in\BB$ for $\alpha>0$, we
	can locally define the smallest decrease
	\begin{align}\label{eq:proof_thm_lyapunov_1}
		-\gamma = \max_{x\in\DD,|x|_\SS\geq \beta} -\Delta V(x)
	\end{align}
	with $\Delta V$ according to~\eqref{eq:def_lyap_func_decrease},
	which is strictly less than zero by assumption (Definition~\ref{def:lyap_func}).
	Furthermore, $\gamma$ is bounded since $\Delta V$ is continuous	and $\DD$
	is compact. Note that~\eqref{eq:proof_thm_lyapunov_1} always has a feasible
	suboptimal solution given by $-\Delta V(x(0))$.
	Forward invariance of $\DD$ allows us to use
	the worst-case decrease from above to obtain
	\begin{align*}
		V(x(k)) & = V(x(k-1)) + V(x(k))-V(x(k-1)) \\
				& = V(x(0)) + \sum_{i=0}^{k-1}V(x(i+1))-V(x(i))\\
				& \leq V(x(0)) - \sum_{i=0}^{k-1}\Delta V(x(i))\\
				& \leq V(x(0)) - (k-1)\gamma.
	\end{align*}
	Since $\gamma >0$ there exists a finite $\bar k$ such that
	$V(x(0)) - (\bar k-1)\gamma <0$ and therefore $V(x(\bar k))<0$,
	yielding a contradiction for any $x(0)\in\DD$.
	We can therefore select $\bar \delta = \max_{x\in\DD} | x |_\SS$
	with $\bar \delta >0$ and $\bar \delta < \infty$ since $\SS\subset\DD$ and
	$\DD$ is compact, proving the desired result.
	\end{proof}

\subsection{Proof of Theorem~\ref{thm:cbf} (Control barrier functions)}
	\label{app:cbf_proof}
	
		We prove this result by showing invariance of $\SS$
		(Property~\ref{item:thm_cbf_1}) first, followed by
		asymptotic stability of $\SS$ (Property~\ref{item:thm_cbf_2}).
	
	\begin{proof}
		Property~\ref{item:thm_cbf_1}:
		From Definition~\ref{def:barrier_function},
		\eqref{eq:def_barrier_function_3} it directly follows
		for any $x(0)\in\SS$ with
		$u(k)=\kappa(x(k))\in K_{\mathrm{CBF}}^2(x)$ that
		$h(x(k))\leq 0$ implies $h(x(k+1))\leq 0$ and therefore
		by induction we have that $x(0)\in\SS$ implies
		$x(k)\in\SS$ for all $k>0$, which proves the desired
		property.
		\par
		Property~\ref{item:thm_cbf_2}: Define~$\mDefFunction{V}{\DD}{\RR}$
		as $V(x)=\max(0, h(x))$, which is continuous due to continuity of
		$h$ and l.p.d. around $\SS$ in $\DD$ according to Definition~\ref{def:pos_def_func}
		by construction of $\SS$. In the following, we show that
		for any control law $\kappa$ with 
		$\kappa(x)\in K_{\mathrm{CBF}}(x)$ according to
		Definition~\ref{def:barrier_function}, it follows that $V$ is a
		Lyapunov function for the closed-loop system~\eqref{eq:system} under
		$u(k)=\kappa(x(k))$ with respect to $\SS$
		in $\DD$. The corresponding decrease can be
		bounded through
		
		\begin{align}\label{eq:lyapunov_decrease}
			\Delta V(x) = \max(0,\min(\Delta h(x), V(x))).
		\end{align}
		Since max/min operations preserve continuity properties,
		it follows that $\Delta V(x)$ is continuous.
		Furthermore, \eqref{eq:lyapunov_decrease} is positive
		definite w.r.t. $\SS$ in $\DD$ since $x\in\SS$ implies
		\begin{align*}
			\max(0,
				\underbrace{
					\min(
					\underbrace{\Delta h(x)}_{\in \RR},
					\underbrace{V(x)}_{=0})
				}_{\leq 0})=0
		\end{align*}
		and $x\in \DD\setminus\SS$ implies
		\begin{align*}
			\max(0,
				\underbrace{
					\min(
					\underbrace{\Delta h(x)}_{>0},
					\underbrace{V(x)}_{>0})
				}_{>0})>0.
		\end{align*}
		Lastly, it remains to verify that the required
		decrease condition, i.e., Property~\ref{item:thm_cbf_2}
		holds. To this end, we distinguish the following three cases,
		which are possible due to invariance of $\SS$ and $\DD$:
		\begin{enumerate}
			\item $x(k)\in\SS \land x(k+1)\in\SS$:
				\begin{align*}
					\underbrace{
						V(x(k+1))}_{=0}
					-
					\underbrace{
						V(x(k))}_{=0} \\
					\leq -
					\max(0,
					\underbrace{
						\min(
						\underbrace{\Delta h(x(k))}_{\in \RR},
						\underbrace{V(x(k))}_{=0})
					}_{\leq 0})=0.
				\end{align*}
			\item $x(k)\in\DD\setminus\SS \land x(k+1)\in \DD\setminus\SS$: Notice that due
				to $h(x(k+1))>0$ we have
				$ h(x(k+1))-h(x(k))\leq -\Delta h(x(k))
					\Rightarrow h(x(k))\geq \Delta h(x(k))$
				in this case, implying
				\begin{align*}
					\underbrace{
						V(x(k+1))}_{=h(x(k+1))}
					-
					\underbrace{
						V(x(k))}_{=h(x(k))} \\
					\leq -
					\max(0,
					\underbrace{
						\min(
						\underbrace{\Delta h(x(k))}_{>0},
						\underbrace{h(x(k))}_{\geq \Delta h(x(k))})
					}_{= \Delta h(x(k))>0})= -\Delta h(x(k)).
				\end{align*}
			\item $x(k)\in\DD\setminus \SS \land x(k+1)\in\SS$: Notice
				that due to $h(x(k+1))\leq 0$ it follows
				$\Delta h(x(k))\geq h(x(k))=V(x(k))$ in this case
				and we have
			\begin{align*}
				\underbrace{
					V(x(k+1))}_{=0}
				-
				V(x(k)) \\
				\leq -
				\max(0,
				\underbrace{
					\min(
					\underbrace{\Delta h(x(k))}_{>0},
					\underbrace{V(x(k))}_{\leq \Delta h(x(k))})
				}_{= V(x(k))>0})\leq -\underbrace{V(x(k))}_{>0}.
			\end{align*}
		\end{enumerate}
		Since all the cases above are guaranteed to be true and satisfy
		Property~\ref{item:thm_cbf_2},
		the desired statement follows directly from Theorem~\ref{thm:lyapunov_stability}.
	\end{proof}
\subsection{Technical lemmas for the proof of Theorem~\ref{thm:BF}}
\label{app:technical_lemmas}

\begin{lemma}\label{lem:existence_bf}
	If the conditions in Theorem~\ref{thm:BF} hold, then the
	minimum~\eqref{eq:opt_BF} exists for all $x\in\RR^n$.
\end{lemma}
\begin{proof}
	The constrained optimization problem~\eqref{eq:opt_BF}
	can equally be written in the condensed form
	\begin{align*}
		\inf_{u_i\in\UU}~&\alpha_f \max(0,h_f(f(..f(f(x,u_0),u_1),..))) \nonumber\\
			& + \sum_{i=0}^{N-1} \mNormGen{\max(0,c(f(..f(f(x,u_0),u_1),..))+\Delta_i\mOnes{},0)},
			% \label{eq:lem_existence_bf}
	\end{align*}
	with $x_{1|k}=f(x,u_{0|k})$, $x_{2|k}=f(f(x,u_{0|k}),u_{1|k})$, .., $\xi_{i|k}=\max(0, c(x_{i|k}))$
	for all $i=0,..,N-1$, and $\xi_{N|k}=\max(0, h_f(x_{N|k}))$. Since compositions, sums, and the maximum of
	continuous functions yield continuous functions and $h_f$, $c$, and $f$ are assumed to be continuous
	on $\RR^n$, it follows that the objective is continuous	on $\RR^n$.	In addition, the input space
	$\UU$ is assumed to be compact, allowing us to apply the Weierstrass
	Extreme Value Theorem~\cite[Proposition A.7]{rawlings2009model}, which implies that the minimum
	exists for all $x\in\RR^n$ and therefore the proof is complete.
\end{proof}
\begin{lemma}\label{lem:bounded_solution}
	If the conditions in Theorem~\ref{thm:BF} hold, then
	it follows that $\DD_\rho = \{x\in\RR^n|h_{\mathrm{PB}}(x)\leq \rho\}\subseteq \XX_0(\mOnes{}\rho)$
	with $\rho>0$ and $\XX_0(\mOnes{}\rho)$ according to~\eqref{eq:def_tight_state_constraints}.
\end{lemma}
\begin{proof}
	For any $x\in\DD_\rho$ it follows that the objective function~\eqref{eq:opt_BF_cost}
	implies that $\mNormGenSmall{\xi_{0|k}^*}\leq \rho$ and it must therefore hold that $x=x_{0|k}\in\XX_0(\mOnes{}\rho)$, which proves
	the desired statement.
\end{proof}
\begin{lemma}\label{lem:compact_bounding_set}
	Let the conditions in Theorem~\ref{thm:BF} hold and consider a state $x\in\DD_{\rho}$ with
	$\DD_\rho = \{x\in\RR^n|h_{\mathrm{PB}}(x)\leq \rho\}$, $\rho \geq 0$
	and input sequence $\{u_{i|k}\}_{i=0}^{N-1}\in\UU^{N-1}$. Define a corresponding state sequence
	$x_{0|k}\mDef x$ and $x_{i+1|k}\mDef f(x_{i|k},u_{i|k})$ for $i=0,..,N-1$ and slack sequence
	$\xi_{i|k}\mDef \max(0, c(x_{i|k})+\Delta_i \mOnes{})$ for $i=0,..,N-1$, $\xi_{N|k}\mDef \max(0, h_f(x_{N|k}))$.
	For every $\rho\geq 0$, there exists a compact set $\ZZ_\rho$ such that
	for all $x\in\DD_\rho$ and $\{u_{i|k}\}_{i=0}^{N-1}\in\UU^{N-1}$ it holds
	that $(\{x_{i|k}\}_{i=0}^N, \{\xi_{i|k}\}_{i=0}^N\})\in\ZZ_\rho$.
\end{lemma}
\begin{proof}
	From Lemma~\ref{lem:bounded_solution}, we know that $x\in\DD_\rho$ will be contained
	in the compact set $\XX_0(\mOnes{}\rho)$. Since the dynamics are continuous and the input
	space is compact, it follows that the prediction mapping of the outer bounding initial set
	$\XX_0(\mOnes{}\rho)$ that contains the states $\{\mXpred_{i|k}\}_{i=0}^N$ will be compact. By noting that a feasible slack sequence for any state sequence $\{\mXpred_{i|k}\}_{i=0}^N$ is given by the continuous mapping
	$\xi_{i|k} = \max(0,c(x_{i|k}))$ and $\xi_{N|k} = \max(0,c(x_{N|k}))$,
	it follows that a valid set of slack sequences corresponding to the compact set of possible states sequences will be compact and therefore the proof is complete.
\end{proof}

\begin{IEEEbiography}[{\includegraphics[width=1in,height=1.25in,clip,keepaspectratio]{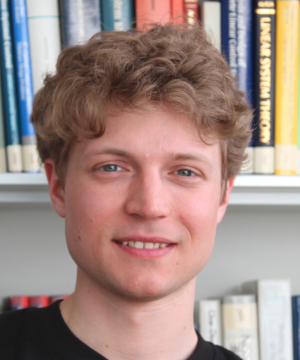}}]{Kim P. Wabersich}
    received a BSc. and MSc. degree in engineering cybernetics from the University of Stuttgart in Germany in 2015 and 2017, respectively. He completed his doctoral studies at ETH Zurich in 2021 and is currently a postdoctoral researcher with the Institute for Dynamic Systems and Control (IDSC) at ETH Zurich. During his studies, he was a research assistant at the Machine Learning and Robotics Lab (University of Stuttgart) and the Daimler Autonomous Driving Research Center (B\"oblingen, Germany and Sunnyvale, CA, USA). His research interests include learning-based model predictive control and safe model-based reinforcement learning.
\end{IEEEbiography}
\begin{IEEEbiography}[{\includegraphics[width=1in,height=1.25in,clip,keepaspectratio]{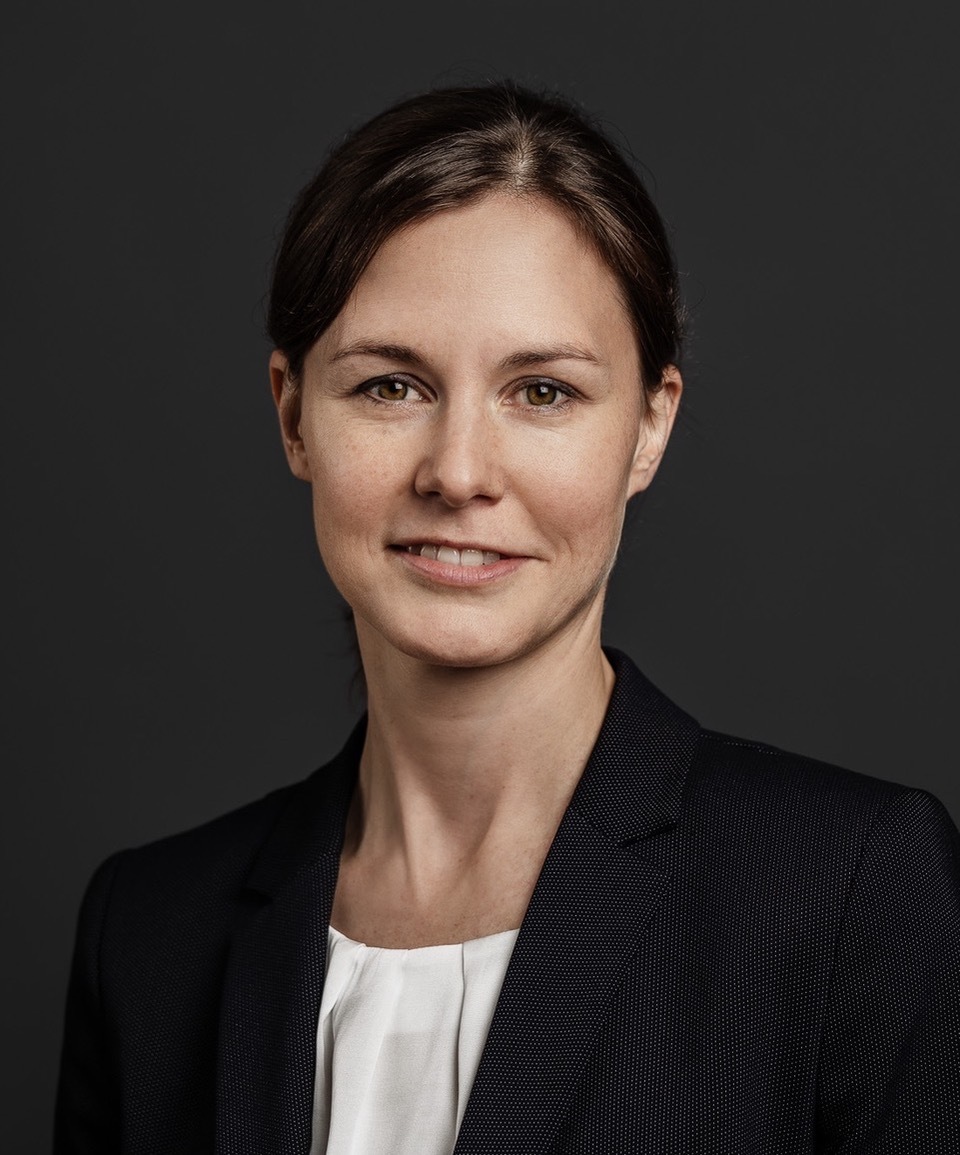}}]{Melanie N. Zeilinger}
    is an Assistant Professor at ETH Zurich, Switzerland. She received the Diploma degree in engineering cybernetics from the University of Stuttgart, Germany, in 2006, and the Ph.D. degree with honors in electrical engineering from ETH Zurich, Switzerland, in 2011. From 2011 to 2012 she was a Postdoctoral Fellow with the Ecole Polytechnique Federale de Lausanne (EPFL), Switzerland. She was a Marie Curie Fellow and Postdoctoral Researcher with the Max Planck Institute for Intelligent Systems, Tübingen, Germany until 2015 and with the Department of Electrical Engineering and Computer Sciences at the University of California at Berkeley, CA, USA, from 2012 to 2014. From 2018 to 2019 she was a professor at the University of Freiburg, Germany. Her current research interests include safe learning-based control, as well as distributed control and optimization, with applications to robotics and human-in-the-loop control.
\end{IEEEbiography}

\end{document}